\documentclass[12pt, final]{amsart}
\usepackage[applemac]{inputenc}
\usepackage{bm}
\usepackage{times}
\usepackage{amsthm}
\usepackage{amsbsy}
\usepackage{amstext}
\usepackage{amssymb}
\usepackage{graphicx}
\usepackage{setspace}
\usepackage{url}
\usepackage{esint}
\usepackage[pdftex,bookmarks=true,colorlinks=true,citecolor=blue,plainpages=false]{hyperref}
\usepackage[round]{natbib}
\PassOptionsToPackage{normalem}{ulem}
\usepackage{ulem}
\newcommand{\mele}[1]{\textcolor{red}{#1}}
\newcommand{\zhu}[1]{\textcolor{blue}{#1}}

\makeatletter

\numberwithin{equation}{section}
\numberwithin{figure}{section}
\theoremstyle{plain}

  \theoremstyle{plain}

\usepackage{amsfonts}
\usepackage{graphicx}
\usepackage{color}
\usepackage{bm}
\usepackage[margin=1in]{geometry}
\newtheorem{theorem}{THEOREM}

\newtheorem{algorithm}{ALGORITHM}
\newtheorem{axiom}{ASSUMPTION}

\newtheorem{corollary}{COROLLARY}

\newtheorem{example}{Example}

\newtheorem{proposition}{PROPOSITION}
\newtheorem{remark}{Remark}

\numberwithin{definition}{section}
\numberwithin{equation}{section}
\numberwithin{axiom}{section}
\numberwithin{lemma}{section}
\numberwithin{theorem}{section}
\numberwithin{proposition}{section}
\numberwithin{example}{section}
\numberwithin{figure}{section}
\numberwithin{table}{section}
\numberwithin{exercise}{section}
\numberwithin{remark}{section}

\makeatother

  \providecommand{\propositionname}{Proposition}
\providecommand{\theoremname}{Theorem}

\begin{document}

\title{Approximate Variational Estimation for a Model of Network Formation}
\author{Angelo Mele}
\address{Carey Business School\\
Johns Hopkins University \\
100 International Dr \\
Baltimore, MD 21202}
\author{Lingjiong Zhu}
\address{Department of Mathematics\\
Florida State University\\
208 Love Building\\
1017 Academic Way\\
Tallahassee, FL 32306}

\thanks{We are grateful to the editor and three excellent referees for their suggestions. We thank Anton Badev, Vincent Boucher, Aureo DePaula, Bryan Graham, Mert G\"{u}rb\"{u}zbalaban, Matt Jackson, Hiro Kaido, Michael Leung, Xiaodong Liu and Demian Pouzo for comments on previous versions of this paper. The second author is partially supported by NSF Grant DMS-1613164.}

\date{First version: June 15, 2015. This version: \today.}
\begin{abstract}
We develop approximate estimation methods for exponential random graph models (ERGMs), whose likelihood  is proportional to an intractable normalizing constant. The usual approach approximates this constant with Monte Carlo simulations, however convergence may be exponentially slow.  We propose a deterministic method, based on a variational mean-field approximation of the ERGM's normalizing constant. We compute lower and upper bounds for the approximation error for any network size, adapting nonlinear large deviations results. This translates into bounds on the distance between true likelihood and mean-field likelihood. Monte Carlo simulations suggest that in practice our deterministic method performs better than our conservative theoretical approximation bounds imply, for a large class of models. 

\textit{Keywords: Networks, Microeconometrics, Large Networks,
Variational Inference, Large deviations, Mean-Field Approximations}
\end{abstract}
\maketitle

\section{Introduction}
 
This paper studies variational mean-field methods to approximate the likelihood of exponential random graph models (ERGMs), a class of statistical network formation models that has become popular in sociology, machine learning, statistics and more recently economics. While a large part of the statistical network literature is devoted to models with unconditionally or conditionally independent links \citep{Graham2014,AiroldiEtAl2008,BickelEtAl2013}, ERGMs allow for conditional and unconditional dependence among links \citep{Snijders2002,WassermanPattison1996}. These models have recently gained attention in economics, because several works have shown that ERGMs have a microeconomic foundation. In fact, the ERGM likelihood naturally emerges as the stationary equilibrium of a potential game, where players engage in a myopic best-response dynamics of link formation \citep{Blume1993, Mele2010a, Badev2013, Chandrasekhar2016, ChandrasekharJackson2012, BoucherMourifie2017}, and in a large class of evolutionary games and social interactions models \citep{Blume1993, DurlaufEtAl2010}. 

Estimation and inference for ERGMs are challenging, because the likelihood of the observed network is proportional to an intractable normalizing constant, that cannot be computed exactly, even in small networks. Therefore, exact Maximum Likelihood estimation (MLE) is infeasible. The usual estimation approach, the Markov Chain Monte Carlo MLE (MCMC-MLE), consists of simulating many networks using the model's conditional link probabilities and approximating the constant and the likelihood with Monte Carlo methods \citep{Snijders2002, Koskinen2004, DiaconisChatterjee2011, Mele2010a}. Estimates of the MCMC-MLE converge almost surely to the MLE if the likelihoods are well-behaved \citep{GeyerThompson1992}. However, a recent literature has shown that the simulation methods used to compute the MCMC-MLE may have exponential slow convergence, making estimation and approximation of the likelihood impractical or infeasible for a large class of ERGMs \citep{BhamidiEtAl2011, DiaconisChatterjee2011, Mele2010a}. An alternative is the Maximum Pseudo-likelihood estimator (MPLE), that finds the parameters that maximize the product of the conditional link probabilities of the model.
While MPLE is simple and computationally fast, the properties of the estimator are not well understood, except in special cases, when some regularity conditions are satisfied \citep{BoucherMourifie2017,Besag1974}; in practice MPLE may give misleading estimates when the dependence among links is strong \citep{GeyerThompson1992}. Furthermore, since the ERGMs are exponential families, networks with the same sufficient statistics will produce the same MLE, but may have different MPLE.

Our work departs from the standard methods of estimation, proposing deterministic approximations of the likelihood, based on the approximated solution of a variational problem. Our strategy is to use a mean-field algorithm to approximate the normalizing constant of the ERGM, at any given parameter value \citep{WainwrightJordan2008,Bishop2006,DiaconisChatterjee2011}. We then maximize the resulting approximate log-likelihood, with respect to the parameters. To be concrete, our approximation consists of using the likelihood of a simpler model with independent links to approximate the constant of the ERGM. The mean-field approximation algorithm finds the likelihood with independent links that minimizes the Kullback-Leibler divergence from the ERGM likelihood. Using this likelihood with independent links, we compute an approximate normalizing constant. We then evaluate the log-likelihood of our model, where the exact normalizing constant is replaced by its mean-field approximation.

Our main contribution is the computation of exact bounds for the approximation error of the normalizing constant's mean-field estimate. Our proofs use the theoretical machinery of \cite{ChatterjeeDembo2014} for non-linear large deviations in models with intractable normalizing constants. Using this powerful tool, we provide explicit lower and upper bounds to the error of approximation for the  mean-field normalizing constant. The bounds depend on the magnitude of the parameters of our model and the size of link externalities \citep{Mele2010a, BoucherMourifie2017, Chandrasekhar2016, DePaula2017}. The result holds for dense and moderately sparse networks.
Remarkably and conveniently the mean-field error  converges to zero as the network becomes large. This guarantees that for large networks,  the log-normalizing constant of an ERGM is well approximated by our mean-field log-normalizing constant. 

The main implication of our main result is that we can compute bounds to the distance between the log-likelihood of the ERGM and our approximate log-likelihood;  these also converge in sup-norm as the network grows large. As a consequence, we can use the approximated likelihood for estimation in large networks. If the likelihood is strictly concave, it is possible to show that our approximate estimator converges to the maximum likelihood estimator as long as the network grows large.  Furthermore, because our bounds may not be sharp, in practice convergence could be faster than what is implied in these results.

While our method is guaranteed to perform well in large graphs,  many applications involve small networks. For example, the school networks data in the National Longitudinal Study of Adoloscent Health (Add Health)  \citep{ BoucherMourifie2017, Moody2001, Badev2013} or the Indian villages in \cite{BanerjeeEtAl2013} include on average about 200-300 nodes. 
To understand the performance of our estimator in practice,  we perform simple Monte Carlo exercises in networks with few hundreds nodes, comparing mean-field estimates to  MCMC-MLE and MPLE. Our Monte Carlo results show that in practice our estimator works better than the theoretical results suggest, for  networks with $50$ to $1000$ nodes. The median mean-field approximation point estimates are close to the true parameters, but exhibit a small bias. Both MCMC-MLE and MPLE show a larger variability of point estimates for the two-stars and triangle parameters, measured as median absolute deviation. When we increase the network size, all three estimators improve, as expected. We conclude that our method's performance is comparable to available estimators in small networks. While our code can be made faster by exploiting efficient matrix algebra libraries and parallelization, the CPU time for estimation is comparable to the estimators implemented in the \texttt{ergm} package in R for networks with less than 200 nodes. 

The main message of our theoretical results and Monte Carlo simulations is that the approximate mean-field approach is a valid alternative to existing methods for estimation of a large class of ERGMs. We note that our theoretical bounds may not be sharp, and in practice the mean-field algorithm may have better performance than what is implied by our conservative results, as confirmed by our Monte Carlo experiments.

To the best of our knowledge, this paper is one of the first works in economics to use mean-field approximations for approximate estimation of complex models. We show that our application of variational approximations has theoretical guarantees, and we can bound the error of approximation. While similar deterministic methods have been used to provide an approximation to the normalizing constant of the ERGM model \citep{DiaconisChatterjee2011, AmirPu2012, Mele2010a, ZhengHe2013, AristoffZhu2014}, we are the first to characterize the variational approximation error for a model with covariates and its computational feasibility.

Our technique can be applied to other models in economics and social sciences. For example, models of social interactions with binary decisions like in \cite{Blume1993}, \cite{Badev2013}, \cite{DurlaufEtAl2010}, models for bundles \citep{FoxLazzati2017}, and models of choices from menus \citep{MisraEtAl2018} have similar likelihoods with intractable normalizing constants . Therefore our method of approximation may allow estimation of these models for large sets of bundles or menu choices.

The rest of the paper is organized as follows. Section 2 presents the theoretical model and variational approximations. Section 3 contains the main theoretical results and the error bounds. Section 4 presents the Monte Carlo results and Section 5 concludes. All the proofs and additional Monte Carlo simulations are in the Appendix.
Additional results and discussions are presented in the Online Appendix.

%%%%%%%%%%%%%%%%%%%%%%%%%%%%%%%%%%%%%%%%%%%%%%%%%%%%%
% THEORETICAL MODEL
%%%%%%%%%%%%%%%%%%%%%%%%%%%%%%%%%%%%%%%%%%%%%%%%%%%%%

\section{Network formation model and variational methods}
\subsection{Exponential random graph models} 
The class of exponential random graphs is an important generative
model for networks and has been extensively used
in applications in many disciplines \citep{WassermanPattison1996, Jackson2008, DePaula2017, Mele2010a, Moody2001, WimmerLewis2010, AmirPu2012}. In this paper we consider a model with nodal covariates, two-stars and triangles. 

Our model assumes that the network consists of $n$ heterogeneous nodes, indexed by $i=1,...,n$; each node is characterized by a $S$-dimensional vector of observed attributes
$\tau_i\in \mathcal{X} := \otimes_{j=1}^{S}\mathcal{X}_{j}$, $i=1,...,n$. 
The sets $\mathcal{X}_{j}$ can represent age, race, gender,
income, etc.\footnote{For instance, if we consider gender and income, 
then $S=2$, and we can take $\otimes_{j=1}^{2}\mathcal{X}_{j}=\{\text{male,female}\}\times\{\text{low, medium, high}\}$.
The sets $\mathcal{X}_{j}$ can be both discrete and continuous. For example, if we consider gender and income, 
we can also take $\otimes_{j=1}^{2}\mathcal{X}_{j}=\{\text{male,female}\}\times\text{[\$50,000,\$200,000]}$. Below we restrict the covariates to be discrete, but we allow the number of types to grow with the size of the network.}
 Let $\alpha$ be a $n\times n$ symmetric matrix  with elements $\alpha_{ij}:=\nu(\tau_i,\tau_j)$, where  $\nu: \mathcal{X} \times \mathcal{X} \rightarrow\mathbb{R}$ is a symmetric function and let $\beta$ and $\gamma$ be scalars. 
For ease of exposition we focus on the case in which the attributes are discrete and finite, but our results hold when this assumption is relaxed and the number of attributes is allowed to increase with the size of the network.

The  likelihood $\pi_n(g,\alpha, \beta, \gamma)$ of observing the adjacency matrix $g$ depends on the composition of links, the number of two-stars and the number of triangles
\begin{equation}
\label{eq:stat_distribution_canonical}
\pi_{n}(g;\alpha,\beta,\gamma)=\frac{\exp\left[Q_{n}(g;\alpha,\beta,\gamma)\right]}
{\sum_{\omega\in\mathcal{G}_n}\exp\left[Q_{n}(\omega;\alpha,\beta,\gamma)\right]},
\end{equation}
where the function $Q$ is called a \emph{potential function} and takes the form
\begin{equation}
\label{eq:potential_triangles}
Q_{n}(g;\alpha,\beta,\gamma)=\sum_{i=1}^{n}\sum_{j=1}^{n}\alpha_{ij}g_{ij}+\frac{\beta}{2n}\sum_{i=1}^{n}\sum_{j=1}^{n}\sum_{k=1}^{n}g_{ij}g_{jk} 
+ \frac{2\gamma}{3n}\sum_{i=1}^{n}\sum_{j=1}^{n}\sum_{k=1}^{n}g_{ij}g_{jk}g_{ki} .
\end{equation}
and $c(\alpha, \beta, \gamma):=\sum_{\omega\in\mathcal{G}_n}\exp\left[Q_{n}(\omega;\alpha,\beta,\gamma)\right]$ is a normalizing constant that guarantees that likelihood \eqref{eq:stat_distribution_canonical} is a proper distribution. 
The second and third term of the potential function \eqref{eq:potential_triangles} are the counts of two-stars and triangles in the network, rescaled by $n$. We rewrite  \eqref{eq:stat_distribution_canonical} as 
\begin{equation}
\label{eq:stat_distribution}
\pi_{n}(g;\alpha,\beta,\gamma)=\exp\left\{ n^{2}\left[T_{n}(g;\alpha,\beta,\gamma)-\psi_{n}(\alpha,\beta,\gamma)\right]\right\},
\end{equation}
where $T_{n}(g;\alpha,\beta,\gamma)= Q_{n}(g;\alpha,\beta,\gamma)n^{-2}$ is the potential scaled by $n^2$ and the log-normalizing constant (scaled by $n^2$) is ,
\begin{equation}\label{psi_n}
\psi_{n}(\alpha,\beta,\gamma)=\frac{1}{n^{2}}\log\sum_{\omega\in\mathcal{G}_n}\exp\left[n^{2}T_{n}(\omega;\alpha,\beta,\gamma)\right],
\end{equation}
and $\mathcal{G}_n:=\{\omega=(\omega_{ij})_{1\leq i,j\leq n}:\omega_{ij}=\omega_{ji}\in\{0,1\}, \omega_{ii}=0, 1\leq i,j\leq n\}$ is the set of all binary matrices with $n$ nodes.
The re-scaling of the potential and the log-normalizing constant is necessary for the asymptotic results, to avoid the explosion of 
the potential function as the size of the network grows large.

\subsection{Microeconomic equilibrium foundations} ERGMs caught the attention of economists because recent works proves a behavioral and equilibrium interpretation of likelihood (\ref{eq:stat_distribution}).\footnote{\cite{Butts2009}, \cite{Mele2010a}, \cite{ChandrasekharJackson2012}, \cite{BoucherMourifie2017}, \cite{Badev2013}, \cite{DePaula2017}.  }
In fact, these likelihoods naturally arise as equilibrium of best-response dynamics in potential games \citep{Blume1993, MondererShapley1996, Butts2009, Mele2010b}. 

To be concrete, let's consider the following game. Players' payoffs are a function of the composition of direct links, friends' popularity and the number of common friends. The utility of network $g$ for player $i$ is given by
\begin{equation}
\label{eq:utility_fcn_triangles}
u_{i}(g,\tau)=\sum_{j=1}^{n}\alpha_{ij}g_{ij}
+\frac{\beta}{n}\sum_{j=1}^{n}\sum_{k=1}^{n}g_{ij}g_{jk} 
+ \frac{\gamma}{n}\sum_{j=1}^{n}\sum_{k=1}^{n}g_{ij}g_{jk}g_{ki},
\end{equation}
Each player forms links with other nodes, maximizing utility \eqref{eq:utility_fcn_triangles}, but taking into account the strategies of other players. We can show that this game of network formation converges to an exponential random graph in a stationary equilibrium, under the following assumptions:\footnote{See \cite{Mele2010a} or \cite{Badev2013} for more technical details and variants of these assumptions. See also \cite{Chandrasekhar2016}, \cite{DePaula2017}, \cite{ChandrasekharJackson2012}, \cite{BoucherMourifie2017}. } (1) the network formation is sequential, with only two active players in each period; (2) two players meet over time with probability $\rho_{ij} := \rho(\tau_i,\tau_j, g_{-ij}) > 0$, where $g_{-ij}$ indicate the network $g$ but link $g_{ij}$; and these meetings are i.i.d. over time; (3) before choosing
whether to form or delete a link, players receive an i.i.d. logistic shock $(\varepsilon_{ij1},\varepsilon_{ij0})$.
 At time $t$, the link $g_{ij}^t$ is formed if
\begin{eqnarray*}
u_i(g^{t}_{ij}=1, g^{t-1}_{-ij},\tau) + u_j(g^{t}_{ij}=1,g^{t-1}_{-ij},\tau) + \varepsilon^{t}_{ij1}\geq 
u_i(g^{t}_{ij}=0, g^{t-1}_{-ij},\tau) + u_j(g^{t}_{ij}=0,g^{t-1}_{-ij},\tau) + \varepsilon^{t}_{ij0} .
\end{eqnarray*}

\cite{Mele2010a} shows that such a model is a potential game \citep{MondererShapley1996} with
potential function given by equation (\ref{eq:potential_triangles}). The probability of observing
network $g$ in the long run is given by (\ref{eq:stat_distribution}) (Theorem 1 in \cite{Mele2010a}), thus  (\ref{eq:stat_distribution}) describes the stationary behavior of the model. In the long-run we observe with high probability the pairwise stable networks, where no pair of players want to form or delete a link.\footnote{In the Online Appendix \ref{section:appendix_microfoundation} we provide more details about the microeconomic foundation of the model for the interested reader.}
\subsection{Variational Approximations}
The constant  $\psi_{n}(\alpha,\beta, \gamma)$ in (\ref{psi_n}) is intractable
because it is a sum over all $2^{\binom{n}{2}}$
possible networks with $n$ nodes; if there are $n=10$ nodes, the sum involves computation of $2^{45}$ potential functions, which is  infeasible.\footnote{See \cite{GeyerThompson1992}, \cite{MurrayEtAl2006},
\cite{Snijders2002} for examples.} In the literature on exponential family likelihoods with intractable normalizing constant, this problem is 
solved by approximating the normalizing constant using Markov Chain Monte Carlo \citep{Snijders2002,Mele2010a,GoodreauEtAl2009, Koskinen2004, CaimoFriel2010, MurrayEtAl2006}. However,
\citet{BhamidiEtAl2011} has shown that such methods
may have exponentially slow convergence for many ERGMs specifications. \\
\indent We propose methods that avoid simulations and we find an approximate likelihood $q_{n}(g)$
 that minimizes the Kullback-Leibler divergence $KL(q_{n}\vert\pi_{n})$
between $q_{n}$ and the true likelihood $\pi_{n}$:
\begin{align}\label{KL:lower:bound}
KL(q_{n}\vert\pi_{n}) & =  \sum_{\omega\in\mathcal{G}_n}q_{n}(\omega)\log\left[\frac{q_{n}(\omega)}{\pi_{n}(\omega;\alpha,\beta)}\right]
\nonumber\\
& =  \sum_{\omega\in\mathcal{G}_n}q_{n}(\omega)
\left[\log q_{n}(\omega)
-n^{2}T_{n}(\omega;\alpha,\beta,\gamma)
+n^{2}\psi_{n}(\alpha,\beta,\gamma)\right] \geq 0.
\end{align}
With some algebra we obtain a lower-bound for the constant
$\psi_{n}(\alpha,\beta,\gamma)$
\begin{equation*}
\psi_{n}(\alpha,\beta,\gamma) \geq  \mathbb{E}_{q_n}\left[T_{n}(\omega;\alpha,\beta,\gamma)\right]+\frac{1}{n^{2}}\mathcal{H}(q_{n}):=\mathcal{L}(q_{n}),
\end{equation*}
where $\mathcal{H}(q_{n})=-\sum_{\omega\in\mathcal{G}_n}q_{n}(\omega)\log q_{n}(\omega)$
is the entropy of distribution $q_{n}$, and
$\mathbb{E}_{q_n}\left[T_{n}(\omega;\alpha,\beta,\gamma)\right]$ is the
expected value of the re-scaled potential, computed according to the distribution $q_n$.  \\
\indent To find the best likelihood approximation we minimize
$KL(q_{n}\vert\pi_{n})$ with respect to $q_{n}$, which is equivalent to
finding the supremum of the lower-bound $\mathcal{L}(q_{n})$, i.e.
\begin{equation}
\psi_{n}(\alpha,\beta,\gamma)=\sup_{q_{n}\in\mathcal{Q}_{n}} \mathcal{L}(q_{n})
= \sup_{q_{n}\in\mathcal{Q}_{n}}\left\{ \mathbb{E}_{q_n}\left[T_{n}(\omega;\alpha,\beta,\gamma)\right]+\frac{1}{n^{2}}\mathcal{H}(q_{n})\right\},
\label{eqn:varapproxgeneral}
\end{equation}
where $\mathcal{Q}_{n}$ is the set of all the probability distributions on $\mathcal{G}_n$.
We have transformed the problem of computing an intractable sum into a
variational problem, i.e. a maximization problem. \\
\indent In general, problem
 \eqref{eqn:varapproxgeneral} has no closed-form solution, thus the
literature suggests to restrict $\mathcal{Q}_{n}$
to be the set of all
completely factorized distribution\footnote{See \citet{WainwrightJordan2008}, \citet{Bishop2006}}
\begin{equation}\label{qn:eqn}
q_{n}(g)=\prod_{i,j}\mu_{ij}^{g_{ij}}(1-\mu_{ij})^{1-g_{ij}},
\end{equation}
where $\mu_{ij}=\mathbb{E}_{q_n}(g_{ij})=\mathbb{P}_{q_n}(g_{ij}=1)$. 
This approximation is called a \emph{mean-field approximation} of the
discrete exponential family. Straightforward
algebra shows that the entropy of $q_{n}$ is additive
\begin{equation*}
\frac{1}{n^{2}}\mathcal{H}(q_{n})
=-\frac{1}{2n^{2}}\sum_{i=1}^{n}\sum_{j=1}^{n}\left[\mu_{ij}\log\mu_{ij}+(1-\mu_{ij})\log(1-\mu_{ij})\right],
\end{equation*}
and the expected potential can be computed as
\begin{equation*}
\mathbb{E}_{q_{n}}\left[T_{n}\left(\omega;\alpha,\beta,\gamma\right)\right]
=  \frac{\sum_{i}\sum_{j}\alpha_{ij}\mu_{ij}}{n^{2}}+\beta\frac{\sum_{i}\sum_{j}\sum_{k}\mu_{ij}\mu_{jk}}{2n^{3}}
+\gamma\frac{2\sum_{i}\sum_{j}\sum_{k}\mu_{ij}\mu_{jk}\mu_{ki}}{3n^{3}}.
\end{equation*}
The mean-field approximation leads to a \emph{lower bound of} $\psi_{n}(\alpha,\beta,\gamma)$, because we restricted $Q_n$,
and the simpler variational
problem is to find a $n\times n$ symmetric matrix $\bm{\mu}(\alpha,\beta,\gamma)$ that solves
\begin{align}
\psi_{n}(\alpha,\beta,\gamma) & \geq \psi_{n}^{MF}(\bm{\mu}(\alpha,\beta,\gamma))
\nonumber
\\
&  = \sup_{\bm{\mu}\in [0,1]^{n^{2}}: \mu_{ij}=\mu_{ji},\forall i,j}
 \bigg\lbrace \frac{1}{n^{2}}\sum_{i,j}\alpha_{ij}\mu_{ij}
 +\frac{\beta}{2n^{3}}\sum_{i,j,k}\mu_{ij}\mu_{jk}
 +\frac{2\gamma}{3n^{3}}\sum_{i,j,k}\mu_{ij}\mu_{jk}\mu_{ki}
 \nonumber
 \\
 & \qquad\qquad- \frac{1}{2n^{2}}\sum_{i=1}^{n}\sum_{j=1}^{n}\left[\mu_{ij}\log\mu_{ij}+(1-\mu_{ij})\log(1-\mu_{ij})\right]\bigg\rbrace.
 \label{eqn:mean-fieldproblem}
\end{align}
The mean-field problem is in general \emph{nonconvex} and
the maximization can be performed using any global optimization method, e.g.
simulated annealing or Nelder-Mead.\footnote{See \cite{WainwrightJordan2008}
and \cite{Bishop2006} for more details.}

%%%%%%%%%%%%%%%%%%%%%%%%%%%%%%%%%%%%%%%%%%%%%%%%%%%
%%%%%%%%%%%%%%%%%%%%%%%%%%%%%%%%%%%%%%%%%%%%%%%%%%%%
%  ASYMPTOTICS AND GRAPH LIMITS
%%%%%%%%%%%%%%%%%%%%%%%%%%%%%%%%%%%%%%%%%%%%%%%%%%%%
\section{Theoretical results}

\subsection{Convergence of the variational mean-field approximation}
For finite $n$, the variational mean-field approximation contains
an error of approximation. In the next theorem, we provide a lower and upper bound to the error of approximation for our model.
\begin{theorem}
\label{thm:mf_bounds}
For fixed network size $n$, the approximation
error of the variational mean-field problem is bounded as
\begin{equation}
\frac{C_{3}(\beta,\gamma)}{n} \leq \psi_{n}(\alpha,\beta,\gamma) - \psi_{n}^{MF}(\bm{\mu}(\alpha,\beta,\gamma)) \leq C_{1}(\alpha,\beta,\gamma)\left(\frac{\log n}{n}\right)^{1/5}+\frac{C_{2}(\alpha,\beta,\gamma)}{n^{1/2}},
\end{equation}
where $C_{1}(\alpha,\beta,\gamma)$, $C_{2}(\alpha,\beta,\gamma)$ are constants depending on $\alpha$, $\beta$ and $\gamma$
and $C_{3}(\beta,\gamma)$ is a constant depending only on $\beta,\gamma$:
\begin{align*}
&C_{1}(\alpha,\beta,\gamma):=c_{1}\cdot\left(\max_{i,j}|\alpha_{ij}|+|\beta|^{4}+|\gamma|^{4}+1\right),
\\
&C_{2}(\alpha,\beta,\gamma):=c_{2}\cdot\left(\max_{i,j}|\alpha_{ij}|+|\beta|+|\gamma|+1\right)^{1/2}\cdot(1+|\beta|^{2}+|\gamma|^{2})^{1/2},
\\
&C_{3}(\beta,\gamma):=|\beta|+4|\gamma|,
\end{align*}
where $c_{1},c_{2}>0$ are some universal constants.
\end{theorem}
%\begin{proof} See Appendix. \end{proof}
The constants in Theorem \ref{thm:mf_bounds} are functions of the parameters $\alpha$, $\beta$ and $\gamma$. The upper bound depends on the maximum payoff from direct links and the intensity of payoff from indirect connections. The lower bound only depends on the strength of indirect connections payoffs (popularity and common friends, that is $\beta$ and $\gamma$). One consequence is that our result holds when the network is dense, but also when it is moderately sparse, in the sense that
$|\alpha_{ij}|$, $|\beta|$ and $|\gamma|$ can have moderate growth in $n$ instead
of being bounded, and the difference of $\psi_{n}$ and $\psi_{n}^{MF}$ goes
to zero if $C_{1}(\alpha,\beta,\gamma)$ grows slower than $n^{1/5}/(\log n)^{1/5}$
and $C_{2}(\alpha,\beta,\gamma)$ grows slower than $n^{1/2}$ as $n\rightarrow\infty$.
For example, if $\max_{i,j}|\alpha_{ij}|=O(n^{\delta_{1}})$, $|\beta|=O(n^{\delta_{2}})$, $|\gamma|=O(n^{\delta_{3}})$
where $\delta_{1}<\frac{1}{5}$ and $\delta_{2},\delta_{3}<\frac{1}{20}$, then $\psi_{n}-\psi_{n}^{MF}$
goes to zero as $n\rightarrow\infty$. On the other hand, if the graph is too sparse, e.g. 
$|\beta|=\Omega(n)$, $|\gamma|=\Omega(n)$, then $\psi_{n}$ cannot be approximated by $\psi_{n}^{MF}$.
%As a by-product of Theorem \ref{thm:mf_bounds}, 
%we will show that $\pi_{n}$ and $q_{n}$ differ 
%from each other with order $n$ in terms of KL divergence,
%which indicates that we may not necessarily be able
%to assume independence among edges in the estimation of
%a large network. 
%
%\begin{proposition}\label{prop:diff}
%Let $q_n^{\ast}(g,\alpha,\beta,\gamma)$ be the probability
%distribution defined in \eqref{qn:eqn}, evaluated at the solution  $\bm{\mu}^{\ast}$ %of variational problem \eqref{eqn:mean-fieldproblem}.
%Then the likelihood of the ERGM model $\pi_n(g,\alpha,\beta,\gamma)$ and 
%the likelihood with independent links $q_n^{\ast}(g,\alpha,\beta,\gamma)$ differ
%in Kullback-Leibler divergence by a term of order $n$.
%\begin{equation}
%KL(q^{\ast}_{n}|\pi_{n})
%\geq
%C_{3}(\beta,\gamma)n.
%\end{equation}
%\end{proposition}
%\begin{proof}
%See Appendix.
%\end{proof}

%The result in Proposition \ref{prop:diff} has practical implications. 
%First, it implies that the ERGM model presented in this work does not 
%necessarily converge to an inhomogeneous Erdos-Renyi model in the large
%$n$ limit, as the likelihoods $\pi_n$ and $q_n^{\ast}$ diverge by a factor of order %$n$.
%Second, while we cannot directly use $q^{\ast}_n$ to approximate
%the likelihood of the ERGM, 
Our main Theorem \ref{thm:mf_bounds} implies
that \emph{we can approximate the log-likelihood of the ERGM} using the 
mean-field approximated constant. 

\begin{proposition}\label{prop:loglik_approx}
Let $\ell_{n}(g_n,\alpha,\beta,\gamma)$ be the log-likelihood of the ERGM
\begin{equation*}
\ell_{n}(g_n,\alpha,\beta,\gamma):=n^{-2}\log\left(\pi_{n}(g_n,\alpha,\beta,\gamma)\right)
=T_{n}(g_n,\alpha,\beta,\gamma)-\psi_{n}(\alpha,\beta,\gamma),
\end{equation*}
and $\ell_{n}^{MF}(g_n,\alpha,\beta,\gamma)$ be the ``mean-field log-likelihood'' obtained by approximating $\psi_n$ with $\psi_n^{MF}$:
\begin{equation*}
\ell_{n}^{MF}(g_n,\alpha,\beta,\gamma):= T_{n}(g_n,\alpha,\beta,\gamma)-\psi_{n}^{MF}(\alpha,\beta,\gamma).
\end{equation*}
Then for any compact  parameter space $\Theta$,
\begin{equation}
0\leq\sup_{\alpha,\beta,\gamma\in \Theta}
\left[\ell_{n}^{MF}-\ell_{n}\right] 
\leq \sup_{\alpha,\beta,\gamma\in \Theta}
C_{1}(\alpha,\beta,\gamma)n^{-1/5}(\log n)^{1/5}+\sup_{\alpha,\beta,\gamma\in \Theta}C_{2}(\alpha,\beta,\gamma)n^{-1/2}.
\end{equation}
\end{proposition}
%\begin{proof}See Appendix.\end{proof}

Proposition \ref{prop:loglik_approx} shows that as the network size grows large, the mean-field approximation of the log-likelihood $\ell_n^{MF}$ is arbitrarily close to the ERGM log-likelihood $\ell_n$. 
This approximation is similar in spirit to the MCMC-MLE method, where the log-normalizing constant is approximated via MCMC to obtain an approximated log-likelihood \citep{GeyerThompson1992, Snijders2002, DePaula2017, MollerWaagepetersen2004}. The main difference is that our approximation is \emph{deterministic} and does not require any simulation.

Note that $\ell_{n}^{MF}=T_{n}-\psi_{n}^{MF}$ 
and $\ell_{n}=T_{n}-\psi_{n}$. 
If $\ell_{n}$ converges to $\ell_{\infty}$ uniformly on 
a compact parameter space $\Theta$, then so does $\ell_{n}^{MF}$.
If $\ell_{n},\ell_{n}^{MF}$ and $\ell_{\infty}$ are continuous
and strictly concave, $\hat{\theta}_{n}$, $\hat{\theta}_{n}^{MF}$, 
the unique maximizers of $\ell_{n}$ and $\ell_{n}^{MF}$ will converge
to the unique maximizer of $\ell_{\infty}$ and hence
$\hat{\theta}_{n}-\hat{\theta}_{n}^{MF}$ will go to zero as $n\rightarrow\infty$.
In the Online Appendix we provide further results on the behavior of the mean-field approximation as $n\rightarrow\infty$, where we discuss the convergence of the log-constant.\footnote{The strict concavity of the likelihood is closely related to the identification of parameters in ERGM models, for which there is a lack of general results (see \cite{Mele2010a}, \cite{DiaconisChatterjee2011}, \cite{AristoffZhu2014} for examples in special cases). }

The result in Proposition \ref{prop:loglik_approx} can be used to bound the distance between the mean-field estimate and the maximum likelihood estimate, for any network size rather than for large $n$. However, such bounds require additional and stronger assumptions on the shape of the likelihood. Indeed, in Appendix~\ref{sec:bound}, we show that a sufficient conditions for computing the bound is a strongly concave likelihood. Under such assumption, we can use the bound in Proposition \ref{prop:loglik_approx}
for the log-likelihood to provide a bound on the distance between MLE and mean-field estimator for any network size $n$. However, these bounds may not be sharp, and therefore we consider them very conservative. In the next section we show via Monte Carlo simulation that in many cases our estimator performs better than the bounds would imply.

%%%%%%%%%%%%%%%%%%%%%%%%%%%%%%%%%%%%%%%%%%%%%%%%
\section{Estimation Experiments}
To understand the performance of the variational approximation in smaller networks, we perform some Monte Carlo experiments. We compare the mean-field approximation with the standard simulation-based MCMC-MLE \cite{GeyerThompson1992, Snijders2002} and the MPLE \citep{Besag1974}. Our method converges in $n^2$ steps, while the MCMC-MLE may converge in $e^{n^2}$ steps. The MPLE usually converges faster.
\subsection{Approximation algorithm for the normalizing constant}
 We implemented our variational approximation for few models in the R package \texttt{mfergm}, available in Github.\footnote{See \url{https://github.com/meleangelo/mfergm}, with instructions for installation and few examples.} We follow the statistical machine learning literature and use an iterative algorithm
that is guaranteed to converge to a local maximum of the mean-field problem \citep{WainwrightJordan2008,Bishop2006}. The algorithm is derived from first-order conditions of the variational mean-field problem. 

Let $\bm{\mu}^{\ast}$  be the matrix that solves the variational problem  \eqref{eqn:mean-fieldproblem}. If we take the derivative with respect to $\mu_{ij}$ and equate to zero, we get
\begin{equation}
\mu_{ij}^{\ast}=\left\lbrace 1+\exp\left[-2\alpha_{ij}-\beta n^{-1}\sum_{k=1}^{n}\left(\mu_{jk}^{\ast}+\mu_{ki}^{\ast}\right)- 4\gamma n^{-1}\sum_{k=1}^{n}\mu_{jk}^{\ast}\mu_{ki}^{\ast} \right]\right\rbrace^{-1}
\label{eq:mean-field_update}
\end{equation}
The logit equation in \eqref{eq:mean-field_update} characterizes a system of equations, whose fixed point is a solution of the mean-field problem. We can therefore start from a matrix $\bm{\mu}$ and iterate the updates in \eqref{eq:mean-field_update} until we reach a fixed point, as described in the following algorithm. 
\begin{algorithm}\label{algo:approx_logconstant} \textbf{Approximation of log-normalizing constant}.
Fix parameters $\alpha,\beta, \gamma$ and a relatively small tolerance value $\epsilon_{tol}$. Initialize the $n\times n$ matrix $\bm{\mu}^{(0)}$ as $ \mu_{ij}^{(0)} \overset{iid}{\sim} U[0,1]$, 
 for all $i,j$. Fix the maximum number of iterations as $T$. 
Then for each $t=0,..., T$:\\
Step 1. Update the entries of matrix $\bm{\mu}^{(t)}$  for all $i,j=1,...,n$
%    \begin{equation}
%\mu_{ij}^{(t+1)}  %=\frac{\exp\left[2\alpha_{ij}+\frac{\beta}{n}\sum_{k=1}^{n}\left(\m%u_{jk}^{(t)}+\mu_{ki}^{(t)}\right) + %\frac{\gamma}{n}\sum_{k=1}^{n}\mu_{jk}^{(t)}\mu_{ki}^{(t)}\right]}
%{1+\exp\left[2\alpha_{ij}+\frac{\beta}{n}\sum_{k=1}^{n}\left(\mu_{j%k}^{(t)}+\mu_{ki}^{(t)}\right)+ %\frac{\gamma}{n}\sum_{k=1}^{n}\mu_{jk}^{(t)}\mu_{ki}^{(t)} %\right]},
%\end{equation}
   \begin{equation}
\mu_{ij}^{(t+1)}  =\left\lbrace 1+\exp\left[-2\alpha_{ij}-\beta n^{-1}\sum_{k=1}^{n}\left(\mu_{jk}^{(t)}+\mu_{ki}^{(t)}\right)- 4\gamma n^{-1}\sum_{k=1}^{n}\mu_{jk}^{(t)}\mu_{ki}^{(t)} \right]\right\rbrace^{-1}.
\end{equation}

Step 2.  Compute the value of the variational mean-field log-constant $\psi_{n}^{MF(t)}$ as
    \begin{align*}
     \psi_{n}^{MF(t)} &= \frac{\sum_{i}\sum_{j}\alpha_{ij}\mu_{ij}^{(t)}}{n^{2}}
 +\beta\frac{\sum_{i}\sum_{j}\sum_{k}\mu_{ij}^{(t)}\mu_{jk}^{(t)}}{2n^{3}}
 +\gamma\frac{2\sum_{i}\sum_{j}\sum_{k}\mu_{ij}^{(t)}\mu_{jk}^{(t)}\mu_{ki}^{(t)}}{3n^{3}}
 \nonumber
 \\
 & \qquad\qquad- \frac{1}{2n^{2}}\sum_{i=1}^{n}\sum_{j=1}^{n}\left[\mu_{ij}^{(t)}\log\mu_{ij}^{(t)}+(1-\mu_{ij}^{(t)})\log(1-\mu_{ij}^{(t)})\right].
    \end{align*}
Step 3. Stop at $t^{\ast}\leq T$ if: $\psi_{n}^{MF(t^{\ast})} - \psi_{n}^{MF(t^{\ast}-1)}\leq \epsilon_{tol}$. 
  Otherwise go back to Step 1.  
\end{algorithm}
The algorithm is initialized at a random uniform matrix $\bm{\mu^{(0)}}$
and iteratively applies the update (\ref{eq:mean-field_update})
to each entry of the matrix, until the increase in the objective function is less than a tolerance level . Since the problem is concave in each
$\mu_{ij}$, this iterative method is guaranteed to find a local
maximum of (\ref{eqn:mean-fieldproblem}).\footnote{There are other alternatives to the random uniform matrix. Indeed a simple starting value could be the set of conditional probabilities of the model at parameters $\alpha, \beta, \gamma$. We did not experiment with this alternative method. }
In our simulations we use a tolerance level of $\epsilon_{tol}=0.0001$. To improve convergence we can re-start the algorithm from different random matrices, as usually done with local optimizers.\footnote{In the Monte Carlo  exercises we have experimented with different numbers of re-starts of the iterative algorithm. However, it is not clear what would be the optimal number of re-starts. A fixed number of
restarts could be suboptimal. It seems reasonable to increase this number as the network grows larger.} This step is
easily parallelizable, thus preserving the order $n^2 $ convergence; while the standard MCMC-MLE is an intrinsically
sequential algorithm and cannot be parallelized.

\subsection{Monte Carlo design} 
All the computations in this section are performed on a PC Dell T6610 with 6 Quad-core Intel i7 (48 threads) and 64GB RAM. 
We test our approximation using $1000$ simulated networks. Each node $i$ has a binary attribute $x_i$, i.e.
$x_i \overset{iid}{\sim}Bernoulli(0.5)$.
Let $z_{ij}=1$ if $x_i=x_j $ and $z_{ij}=0$ otherwise. 
\begin{align}
&t_z(g):=\frac{1}{n^2}\sum_{i=1}^n \sum_{j=1}^n g_{ij}  z_{ij};\hspace{.1cm}t_{-z}(g):= \frac{1}{n^2}\sum_{i=1}^n \sum_{j=1}^n  g_{ij} (1-z_{ij}),\\
&t_e(g):=\frac{1}{n^2}\sum_{i=1}^n \sum_{j=1}^n g_{ij};\hspace{.1cm} t_s(g):=\frac{1}{n^3}  \sum_{i=1}^n \sum_{j=1}^n \sum_{k=1}^n g_{ij}g_{jk}  ;\hspace{.1cm} t_t(g):= \frac{1}{n^3}\sum_{i=1}^n \sum_{j=1}^n \sum_{k=1}^n g_{ij}g_{jk}g_{ki},\nonumber
\end{align}
where $t_e(g)$, $t_s(g)$ and $t_t(g)$ are the fraction of links, two-stars and triangles respectively. And $t_z(g)$ and $t_{-z}(g)$ are the fractions of links of the same type and different type, respectively.
The log-likelihood of the model $\ell_n (g;\alpha,\beta,\gamma)$ is
\begin{equation}
    \ell_n (g,x;\alpha,\beta,\gamma) = \alpha_1 t_z(g) + \alpha_2 t_{-z}(g) + (\beta/2) t_s(g) + (2\gamma/3) t_t(g) -\psi_n(\alpha_1, \alpha_2,\beta, \gamma). 
    \label{eq:loglik_simulations}
\end{equation}

%\begin{align}%\label{eq:loglik_simulations}
%    \ell_n (g,&x;\alpha,\beta,\gamma) = \alpha_1 %\left[\frac{1}{n^2}\sum_{i=1}^n \sum_{j=1}^n g_{ij}  %z_{ij}\right] + \alpha_2 \left[\frac{1}{n^2}\sum_{i=1}^n %\sum_{j=1}^n  g_{ij} (1-z_{ij})\right] \\
%    &+ \frac{\beta}{2} \left[ \frac{1}{n^3}  \sum_{i=1}^n %\sum_{j=1}^n \sum_{k=1}^n g_{ij}g_{jk} \right] + %\frac{\gamma}{6} \left[ \frac{1}{n^3}\sum_{i=1}^n %\sum_{j=1}^n \sum_{k=1}^n g_{ij}g_{jk}g_{ki}\right] %-\psi_n(\alpha_1, \alpha_2,\beta, \gamma). \nonumber
%\end{align}
%The first term counts the fraction of links in which $i$ %and $j$ have the same value of $x$, while the second term %is the fraction of links in which $i$ and $j$ have %different values of $x$. The third term is the number of %stars, while the third is the number of triangles. 
For computational convenience  we rewrite model  \eqref{eq:loglik_simulations} in a slightly different but equivalent way
%\begin{align}\label{eq:loglik_simulations_R}
%    \ell_n (g,&x;\tilde{\alpha},\beta,\gamma) = \tilde{\alpha}_1 \left[\frac{1}{n^2}\sum_{i=1}^n \sum_{j=1}^n g_{ij}  \right] + \tilde{\alpha}_2 \left[\frac{1}{n^2}\sum_{i=1}^n \sum_{j=1}^n  g_{ij} z_{ij}\right] \\
%    &+ \frac{\beta}{2} \left[ \frac{1}{n^3}  \sum_{i=1}^n \sum_{j=1}^n \sum_{k=1}^n g_{ij}g_{jk} \right] + \frac{\gamma}{6} \left[ \frac{1}{n^3}\sum_{i=1}^n \sum_{j=1}^n \sum_{k=1}^n g_{ij}g_{jk}g_{ki}\right] - \psi_n(\alpha_1, \alpha_2,\beta, \gamma), \nonumber
%\end{align}
\begin{equation}
        \ell_n (g,x;\tilde{\alpha},\beta,\gamma) = \tilde{\alpha}_1 t_e(g) + \tilde{\alpha}_2 t_z(g)
    + (\beta/2) t_s(g) + (2\gamma/3) t_t(g) - \psi_n(\alpha_1, \alpha_2,\beta, \gamma),
    \label{eq:loglik_simulations_R}
\end{equation}
where we have defined $\tilde{\alpha}_1:=\alpha_2$ and $\tilde{\alpha}_2 := \alpha_1 - \alpha_2$. We use specification \eqref{eq:loglik_simulations_R} in our simulations.\footnote{There are other small differences in how we have specified the model and how  we have setup computations using the \texttt{statnet} package in R, that can affect the comparability of the simulation results, in particular the normalizations of the sufficient statistics.  This is handled by our \texttt{mfergm} package, to guarantee comparability of the estimates obtained with MCMC-MLE, MPLE and Mean-field approximate inference.}

To generate the artificial networks, we draw i.i.d. attributes $x_i \sim Bernoulli(0.5)$, initialize a network with $n$ nodes as an Erdos-Renyi graph with probability $p=e^{\tilde{\alpha}_1}/(1+e^{\tilde{\alpha}_1})$, and  then run the Metropolis-Hastings network sampler using the \texttt{simulate.ergm} command in the R package \texttt{ergm} to sample $1000$ networks, each separated by $10,000$ iterations, and after a burn-in of $10$ million iterations.\footnote{The code is available in the Github package \texttt{mfergm}, and the function is \texttt{simulate.model$\#$}, where $\#$ stands for the model number: 2 is the model with $\gamma=0$, 3 is the model with $\beta=0$, and 4 is the model with $\beta\neq 0 $ and $\gamma \neq 0$.} The MCMC-MLE estimator is solved using the Stochastic approximation method of \cite{Snijders2002}, where each simulation has a burnin of $100,000$ iterations of the Metropolis-Hastings sampler and networks are sampled every $1000$ iterations. The other convergence parameters are kept at default of the \texttt{ergm} package. The MPLE estimate is obtained using the default parameters in \texttt{ergm}. To be sure that our results do not depend on the initialization of the parameters, we start each estimator at the true parameter value, thus decreasing the computational time required for convergence. All the code is available in Github for replication.

\begin{table}[!ht]
\caption{Monte Carlo estimates, comparison of three methods. True parameter vector is $(\tilde{\alpha}_1, \tilde{\alpha}_2, \beta,\gamma)=(-2,1,1,1)$}
\centering
\begin{small}
\begin{tabular}{lcccccccccccc}
  \hline
$n=50$ & \multicolumn{4}{c}{MCMC-MLE} & \multicolumn{4}{c}{MEAN-FIELD} & \multicolumn{4}{c}{MPLE} \\
  \hline
 & $\tilde{\alpha}_1$ & $\tilde{\alpha}_2$ & $\beta$ & $\gamma  $ & $\tilde{\alpha}_1$ & $\tilde{\alpha}_2$ & $\beta$  & $\gamma$ & $\tilde{\alpha}_1$ & $\tilde{\alpha}_2$ & $\beta$  & $\gamma $  \\
  \hline
median & -2.002 & 1.024 & 0.716 & -2.042 & -2.000 & 0.998 & 1.000 & 0.999 & -1.957 & 1.016 & 0.118 & -0.584 \\ 
  mad & 0.295 & 0.238 & 3.412 & 26.132 & 0.044 & 0.040 & 0.012 & 0.012 & 0.268 & 0.179 & 3.261 & 16.540 \\ 
 \hline\hline

$n=100$ & \multicolumn{4}{c}{MCMC-MLE} & \multicolumn{4}{c}{MEAN-FIELD} & \multicolumn{4}{c}{MPLE} \\
  \hline
median & -1.991 & 0.991 & 0.886 & 1.183 & -2.002 & 0.995 & 1.001 & 0.999 & -1.974 & 0.991 & 0.713 & 1.020 \\ 
  mad & 0.197 & 0.117 & 2.324 & 16.150 & 0.020 & 0.017 & 0.005 & 0.005 & 0.178 & 0.085 & 2.237 & 10.478 \\ 
\hline\hline

$n=200$& \multicolumn{4}{c}{MCMC-MLE} & \multicolumn{4}{c}{MEAN-FIELD} & \multicolumn{4}{c}{MPLE} \\
  \hline
median & -2.000 & 1.000 & 1.043 & 0.438 & -2.003 & 0.995 & 1.001 & 0.999 & -1.990 & 1.000 & 0.853 & 0.657 \\ 
  mad & 0.127 & 0.064 & 1.686 & 10.627 & 0.009 & 0.009 & 0.002 & 0.002 & 0.125 & 0.046 & 1.613 & 7.950 \\ 
 \hline\hline

$n=500$& \multicolumn{4}{c}{MCMC-MLE} & \multicolumn{4}{c}{MEAN-FIELD} & \multicolumn{4}{c}{MPLE} \\
  \hline
median & -2.000 & 1.001 & 1.000 & 0.706 & -2.002 & 0.994 & 1.016 & 0.992 & -1.994 & 1.001 & 0.912 & 0.762 \\ 
  mad & 0.084 & 0.033 & 1.090 & 6.962 & 0.007 & 0.008 & 0.023 & 0.011 & 0.074 & 0.023 & 0.945 & 4.691 \\ 
   \hline\hline

%$n=1000$& \multicolumn{4}{c}{MCMC-MLE} & \multicolumn{4}{c}{MEAN-FIELD} & \multicolumn{4}{c}{MPLE} \\
%  \hline
% & $\tilde{\alpha}_1$ & $\tilde{\alpha}_2$ & $\beta$ & $\gamma  $ & $\tilde{\alpha}_1$ & $\tilde{\alpha}_2$ & $\beta$  & %$\gamma$ & $\tilde{\alpha}_1$ & $\tilde{\alpha}_2$ & $\beta$  & $\gamma $  \\
%   \hline
%*  median & -1.976 & 0.986 & 1.578 & 0.483 & -1.844 & 0.916 & 5.223 & -1.000 & -1.970 & 0.994 & 1.551 & -0.960 \\ 
%  0.05 & -2.191 & 0.894 & -0.231 & -13.421 & -1.978 & 0.812 & 4.321 & -1.084 & -2.173 & 0.924 & -0.822 & -11.144 \\ 
%  0.95 & -1.805 & 1.089 & 3.920 & 14.010 & -1.783 & 1.033 & 6.900 & -0.920 & -1.736 & 1.076 & 3.857 & 11.229 \medskip\\ 
%   \hline\hline
\end{tabular}
\end{small}
\flushleft Results of 1000 Monte Carlo estimates using three methods. MCMC-MLE is the Monte Carlo Maximum Likelihood estimator of \cite{GeyerThompson1992}, as implemented in \texttt{ergm} in \texttt{R}, with a stochastic approximation algorithm \cite{Snijders2002}. MEAN-FIELD is our method. MPLE is the Maximum Pseudo-Likelihood Estimate. Each network is generated with a 10 million run of the Metropolis-Hastings sampler of the \texttt{ergm} command in R, sampling every 10000 iterations. mad is the median absolute deviation.
\label{tab:MC_-2_1_1_1_twostar_triangles}
\end{table}

\subsection{Results}
The first model has true parameter vector $(\tilde{\alpha}_1, \tilde{\alpha}_2, \beta,\gamma)=(-2,1,1,1)$ and the summaries of point estimates are reported in Table \ref{tab:MC_-2_1_1_1_twostar_triangles}. We show results for  $n=50, 100,200$ and $500$; reporting median and median absolute deviation (mad) of point estimates for each parameter. 

The median estimates of the mean-field approximation are quite stable and exhibit a small bias, as is well known in the literature \citep{WainwrightJordan2008, Bishop2006}. The median results for MCMC-MLE and MPLE are quite precise for $\tilde{\alpha}_1$ and $\tilde{\alpha}_2$, but vary a lot for $\beta$ and $\gamma$, as shown by the large median absolute deviation. Nonetheless the median point estimates of $\beta$ and $\gamma$  are slowly converging to the true parameter vector as $n$ increases.\footnote{Some of the bias in the mean-field approximation may be due to the fact that we only initialize $\bm{\mu}$ once in these simulations.} Therefore, the mean-field approximation provides estimates in line with MPLE and MCMC-MLE, with more reliability for $\beta$ and $\gamma$ in these small sample estimation exercises.

\begin{table}[!ht]
\caption{Monte Carlo estimates, comparison of three methods. True parameter vector is $(\tilde{\alpha}_1, \tilde{\alpha}_2, \beta,\gamma)=(-3,2,1,3)$}
\centering
\begin{small}
\begin{tabular}{lcccccccccccc}
  \hline
$n=50$ & \multicolumn{4}{c}{MCMC-MLE} & \multicolumn{4}{c}{MEAN-FIELD} & \multicolumn{4}{c}{MPLE} \\
  \hline
 & $\tilde{\alpha}_1$ & $\tilde{\alpha}_2$ & $\beta$ & $\gamma  $ & $\tilde{\alpha}_1$ & $\tilde{\alpha}_2$ & $\beta$  & $\gamma$ & $\tilde{\alpha}_1$ & $\tilde{\alpha}_2$ & $\beta$  & $\gamma $  \\
  \hline
median & -3.041 & 2.064 & 0.743 & -0.512 & -3.007 & 1.993 & 1.000 & 3.000 & -3.026 & 2.083 & 0.215 & 1.764 \\ 
  mad & 0.476 & 0.424 & 3.811 & 25.109 & 0.026 & 0.026 & 0.013 & 0.014 & 0.514 & 0.401 & 3.593 & 16.538 \\ 
\hline\hline

$n=100$ & \multicolumn{4}{c}{MCMC-MLE} & \multicolumn{4}{c}{MEAN-FIELD} & \multicolumn{4}{c}{MPLE} \\
  \hline
median & -3.006 & 2.015 & 0.932 & 0.587 & -3.011 & 1.989 & 1.000 & 2.999 & -2.991 & 2.018 & 0.682 & 1.773 \\ 
  mad & 0.261 & 0.206 & 2.538 & 17.905 & 0.016 & 0.016 & 0.008 & 0.008 & 0.259 & 0.194 & 2.364 & 12.123 \\ 
  \hline\hline

$n=200$& \multicolumn{4}{c}{MCMC-MLE} & \multicolumn{4}{c}{MEAN-FIELD} & \multicolumn{4}{c}{MPLE} \\
  \hline
median & -3.012 & 2.007 & 1.069 & 2.807 & -3.011 & 1.988 & 1.000 & 2.999 & -3.005 & 2.011 & 0.932 & 2.988 \\ 
  mad & 0.158 & 0.117 & 1.822 & 11.360 & 0.008 & 0.008 & 0.004 & 0.004 & 0.156 & 0.109 & 1.714 & 8.144 \\ 
   \hline\hline

$n=500$& \multicolumn{4}{c}{MCMC-MLE} & \multicolumn{4}{c}{MEAN-FIELD} & \multicolumn{4}{c}{MPLE} \\
  \hline
median & -2.998 & 2.000 & 0.951 & 3.047 & -3.011 & 1.988 & 1.002 & 2.999 & -2.998 & 2.001 & 0.921 & 3.117 \\ 
  mad & 0.096 & 0.061 & 1.276 & 7.191 & 0.003 & 0.003 & 0.002 & 0.002 & 0.083 & 0.049 & 1.077 & 5.378 \\ 
 \hline\hline

%$n=1000$& \multicolumn{4}{c}{MCMC-MLE} & \multicolumn{4}{c}{MEAN-FIELD} & \multicolumn{4}{c}{MPLE} \\
%  \hline
% & $\tilde{\alpha}_1$ & $\tilde{\alpha}_2$ & $\beta$ & $\gamma  $ & $\tilde{\alpha}_1$ & $\tilde{\alpha}_2$ & $\beta$  & %$\gamma$ & $\tilde{\alpha}_1$ & $\tilde{\alpha}_2$ & $\beta$  & $\gamma $  \\
%   \hline
%*  median & -1.976 & 0.986 & 1.578 & 0.483 & -1.844 & 0.916 & 5.223 & -1.000 & -1.970 & 0.994 & 1.551 & -0.960 \\ 
%  0.05 & -2.191 & 0.894 & -0.231 & -13.421 & -1.978 & 0.812 & 4.321 & -1.084 & -2.173 & 0.924 & -0.822 & -11.144 \\ 
%  0.95 & -1.805 & 1.089 & 3.920 & 14.010 & -1.783 & 1.033 & 6.900 & -0.920 & -1.736 & 1.076 & 3.857 & 11.229 \medskip\\ 
%   \hline\hline
\end{tabular}
\end{small}
\flushleft Notes: see notes for Table \ref{tab:MC_-2_1_1_1_twostar_triangles}.
\label{tab:MC_-3_2_1_3_twostar_triangles}
\end{table}

\begin{table}[!ht]
\caption{Monte Carlo estimates, comparison of three methods. True parameter vector is $(\tilde{\alpha}_1, \tilde{\alpha}_2, \beta,\gamma)=(-3,1,2,1)$}
\centering
\begin{small}
\begin{tabular}{lcccccccccccc}
  \hline
$n=500$ & \multicolumn{4}{c}{MCMC-MLE} & \multicolumn{4}{c}{MEAN-FIELD} & \multicolumn{4}{c}{MPLE} \\
  \hline
 & $\tilde{\alpha}_1$ & $\tilde{\alpha}_2$ & $\beta$ & $\gamma  $ & $\tilde{\alpha}_1$ & $\tilde{\alpha}_2$ & $\beta$  & $\gamma$ & $\tilde{\alpha}_1$ & $\tilde{\alpha}_2$ & $\beta$  & $\gamma $  \\
  \hline
median & -3.001 & 0.998 & 2.028 & -19.034 & -3.000 & 1.000 & 2.000 & 1.000 & -2.996 & 1.000 & 1.488 & -7.923 \\ 
  mad & 0.086 & 0.065 & 7.205 & 165.600 & 0.011 & 0.011 & 0.0001 & 0.0001 & 0.078 & 0.044 & 6.345 & 84.681 \\ 

  \hline
$n=1000$ & \multicolumn{4}{c}{MCMC-MLE} & \multicolumn{4}{c}{MEAN-FIELD} & \multicolumn{4}{c}{MPLE} \\
  \hline
 & $\tilde{\alpha}_1$ & $\tilde{\alpha}_2$ & $\beta$ & $\gamma  $ & $\tilde{\alpha}_1$ & $\tilde{\alpha}_2$ & $\beta$  & $\gamma$ & $\tilde{\alpha}_1$ & $\tilde{\alpha}_2$ & $\beta$  & $\gamma $  \\
  \hline
median & -2.999 & 1.004 & 1.809 & -0.716 & -3.000 & 1.000 & 2.000 & 1.000 & -2.999 & 1.002 & 1.757 & 0.540 \\ 
  mad & 0.057 & 0.037 & 4.891 & 125.293 & 0.005 & 0.005 & 0.0001 & 0.0001 & 0.049 & 0.022 & 4.113 & 61.328 \\ 
  \hline\hline
\end{tabular}
\end{small}
\flushleft Notes: see notes for Table \ref{tab:MC_-2_1_1_1_twostar_triangles}. The case with $n=1000$ contains only $500$ monte carlo replications.
\label{tab:MC_-3_1_2_1_twostar_triangles}
\end{table}

The second set of results is for a model with parameters $(\tilde{\alpha}_1, \tilde{\alpha}_2, \beta,\gamma)=(-3,2,1,3)$, see Table \ref{tab:MC_-3_2_1_3_twostar_triangles}. The pattern is similar to Table \ref{tab:MC_-2_1_1_1_twostar_triangles}. Indeed the mean-field estimator seems to work relatively well in most cases, especially for the estimates of $\beta$ and $\gamma$.
For parameters $\tilde{\alpha}_1,\tilde{\alpha}_2$ our mean-field estimator (median) bias persists as $n$ increases. 
Finally, we also report a simulation with a larger network with $n=500,1000$ in Table \ref{tab:MC_-3_1_2_1_twostar_triangles}. The results are the same as the other tables and the mean-field approximation is robustly close to the true parameter values in most simulations. 

These Monte Carlo experiments suggests that our approximation method performs well in practice. We conclude that in most cases the mean-field approximation algorithm works better than our conservative theoretical results suggest.\footnote{While these results are encouraging, in Appendix we report some example of non-convergence of the mean-field algorithm, mostly due to our iterative algorithm getting trapped in a local maximum in some simulations.}

% \subsection{Computational speed}
% The computational speed of the three estimators is similar for small networks. For $n=100$, the mean-field approximation takes about $3.5s$ to estimate the model, while an MCMC-MLE with a burnin of $100,000$ and sampling every $1000$ iterations takes approximately $5.5s$ and the MPLE takes about $1.7s$. For $n=50$ the estimates take $1.6s$ for mean-field, $4s$ for MCMC-MLE and $1.2s$ for MPLE. 

%%%%%%%%%%%%%%%%%%%%%%%%%%%%%%%%%%%%%%%%%%%%%%%%%%%%%%%%%%%
\section{Conclusions and Future work}
We have shown that for a large class of exponential random graph models (ERGM), we can
approximate the normalizing constant of the likelihood using a mean-field variational approximation algorithm \citep{WainwrightJordan2008, Bishop2006, DiaconisChatterjee2011, Mele2010a}. Our theoretical results use nonlinear large deviations methods \citep{ChatterjeeDembo2014} to bound the error of approximation, showing that it converges to zero as the network grows. 

Our estimation method consists of replacing the log-normalizing constant in the log-likelihood of the ERGM with the value approximated by the mean-field algorithm; we then find the parameters that maximize such approximate log-likelihood. Since our approximated constant converges to the true constant in large networks, the approximate log-likelihood converges to the correct log-likelihood in sup-norm, as the network becomes large. If the likelihoods are well-behaved and not too flat around the maximizers, we can also show that our estimate converges to MLE.

Using an iterative procedure to find the approximate mean-field constant, we compare our method to MCMC-MLE and MPLE \citep{Snijders2002, Boucher2013, Besag1974, DePaula2017} in a simple Monte Carlo study for small networks. The mean-field approximation exhibits a small bias, but the median estimates are similar to MCMC-MLE and MPLE. Theoretically, our method converges in a number of steps proportional to the number of potential links of a network, while MCMC-MLE could be exponentially slow.

While these results are encouraging, there are several open problems and possible research directions. First, it is not clear that the mean-field estimates are consistent. Our small Monte Carlo seem to indicate that there is a persistent bias term, but there is no general proof in this setting along the lines of \cite{BickelEtAl2013} for stochastic block models. Second, it is not clear that the ERGM model is identified for all parameter values. Indeed some results in this literature suggest otherwise \citep{DiaconisChatterjee2011, Mele2010a, BoucherMourifie2017}. A promising research avenue for the future is the use of the large $n$ mean-field approximation to understand identification, similarly to what has been done with graph limits in \cite{DiaconisChatterjee2011}. Third, while the mean-field approximation is simple and we are able to compute the approximation errors, our lower and upper bounds may not be sharp. This raises the question of whether there is another factorization (like in structured mean-field) that leads to better approximations and faster convergence \citep{XingJordanRussell2002}. We hope that our work will stimulate additional research and more applications of this class of approximations.

%%%%%%%%%%%%%%%%%%%%%%%%%%%%%%%%%%%%%%%%%%%%%%%%%%%%%%%%%%

%%%%%%%%%%%%%%%%%%%%%%%%%%%%%%%%%%%%%%%%%%%%%%%%%%%%%%%%%%
%%%%%%%%%%%%%%%%%%%%%%%%%%%%%%%%%%%%%%%%%%%%%%%%%%%%%%%%%%
%%%%%%%%%%%%%%%%%%%%%%%%%%%%%%%%%%%%%%%%%%%%%%%%%%%%%%%%%%
%%%%%%%%%%%%%%%%%%%%%%%%%%%%%%%%%%%%%%%%%%%%%%%%%%%%%%%%%%
%%%%%%%%%%%%%%%%%%%%%%%%%%%%%%%%%%%%%%%%%%%%%%%%%%%%%%%%%%
%%%%%%%%%%%%%%%%%%%%%%%%%%%%%%%%%%%%%%%%%%%%%%%%%%%%%%%%%%
%%%%%%%%%%%%%%%%%%%%%%%%%%%%%%%%%%%%%%%%%%%%%%%%%%%%%%%%%%
%%%%%%%%%%%%%%%%%%%%%%%%%%%%%%%%%%%%%%%%%%%%%%%%%%%%%%%%%%
%%%%%%%%%%%%%%%%%%%%%%%%%%%%%%%%%%%%%%%%%%%%%%%%%%%%%%%%%%

\newpage

\bibliographystyle{jmr}
\bibliography{thesisbib}

\newpage

% Appendix
\appendix
\section*{APPENDIX}
\setcounter{section}{1}

\subsection{Proof of Theorem \ref{thm:mf_bounds}}
In this proof we will try to follow closely the notation in \citet{ChatterjeeDembo2014}.
Suppose that $f:[0,1]^{N}\rightarrow\mathbb{R}$ is twice continuously
differentiable in $(0,1)^{N}$, so that $f$ and all its first and
second order derivatives extend continuously to the boundary. Let
$\Vert f\Vert$ denote the supremum norm of $f:[0,1]^{N}\rightarrow\mathbb{R}$.
For each $i$ and $j$, denote
\begin{equation}
f_{i}:=\frac{\partial f}{\partial x_{i}},\qquad f_{ij}:=\frac{\partial^{2}f}{\partial x_{i}\partial x_{j}},
\end{equation}
and let
\begin{equation}
a:=\Vert f\Vert,\qquad b_{i}:=\Vert f_{i}\Vert,\qquad c_{ij}:=\Vert f_{ij}\Vert.
\end{equation}
Given $\epsilon>0$, $\mathcal{D}(\epsilon)$ is the finite subset
of $\mathbb{R}^{N}$ so that for any $x\in\{0,1\}^{N}$, there exists
$d=(d_{1},\ldots,d_{N})\in\mathcal{D}(\epsilon)$ such that
\begin{equation}
\sum_{i=1}^{N}(f_{i}(x)-d_{i})^{2}\leq N\epsilon^{2}.\label{fcriterion}
\end{equation}

Let us define
\begin{equation}
F:=\log\sum_{x\in\{0,1\}^{N}}e^{f(x)},
\end{equation}
and for any $x=(x_{1},\ldots,x_{N})\in[0,1]^{N}$,
\begin{equation}
I(x):=\sum_{i=1}^{N}[x_{i}\log x_{i}+(1-x_{i})\log(1-x_{i})].
\end{equation}

In the proof we rely on Theorem 1.5 in \citet{ChatterjeeDembo2014} that we reproduce in Theorem \ref{thm:theoremChatterjeeDembo2014} to help the reader:

\begin{theorem}[\citet{ChatterjeeDembo2014}] For any $\epsilon>0$,
\begin{equation}
\sup_{x\in[0,1]^{N}}\{f(x)-I(x)\}-\frac{1}{2}\sum_{i=1}^{N}c_{ii}\leq F\leq\sup_{x\in[0,1]^{N}}\{f(x)-I(x)\}+\mathcal{E}_{1}+\mathcal{E}_{2},
\end{equation}
where
\begin{equation}\label{eqn:E1}
\mathcal{E}_{1}:=\frac{1}{4}\left(N\sum_{i=1}^{N}b_{i}^{2}\right)^{1/2}\epsilon+3N\epsilon+\log|\mathcal{D}(\epsilon)|,
\end{equation}
and
\begin{eqnarray}\label{eqn:E2}
\mathcal{E}_{2} & :=4\left(\sum_{i=1}^{N}(ac_{ii}+b_{i}^{2})+\frac{1}{4}\sum_{i,j=1}^{N}(ac_{ij}^{2}+b_{i}b_{j}c_{ij}+4b_{i}c_{ij})\right)^{1/2}\\
 & \qquad+\frac{1}{4}\left(\sum_{i=1}^{N}b_{i}^{2}\right)^{1/2}\left(\sum_{i=1}^{N}c_{ii}^{2}\right)^{1/2}+3\sum_{i=1}^{N}c_{ii}+\log2.\nonumber
\end{eqnarray}
\label{thm:theoremChatterjeeDembo2014}
\end{theorem}

We will use the Theorem \ref{thm:theoremChatterjeeDembo2014} to derive the lower and upper bound of the mean-field approximation problem. 
Our results extend Theorem 1.7. in \cite{ChatterjeeDembo2014}
from the ERGM with two-stars and triangles
to the model that allows nodal covariates.
Notice that in our
case the $N$ of the theorem is the number of links, i.e. $N = \binom{n}{2}$.
Let
\begin{equation}
Z_{n}:=\sum_{x_{ij}\in\{0,1\},x_{ij}=x_{ji},1\leq i<j\leq n}e^{\sum_{1\leq i,j\leq n}\alpha_{ij}x_{ij}+\frac{\beta}{2n}\sum_{1\leq i,j,k\leq n}x_{ij}x_{jk}
+\frac{2\gamma}{3n}\sum_{1\leq i,j,k\leq n}x_{ij}x_{jk}x_{ki}},
\end{equation}
be the normalizing factor and also define
\begin{align}
L_{n} & :=\sup_{x_{ij}\in[0,1],x_{ij}=x_{ji},1\leq i<j\leq n}\bigg\{\frac{1}{n^{2}}\sum_{i,j}\alpha_{ij}x_{ij}+\frac{\beta}{2n^{3}}\sum_{i,j,k}x_{ij}x_{jk}
+\frac{2\gamma}{3n^{3}}\sum_{i,j,k}x_{ij}x_{jk}x_{ki}\\
 & \qquad-\frac{1}{n^{2}}\sum_{1\leq i<j\leq n}[x_{ij}\log x_{ij}+(1-x_{ij})\log(1-x_{ij})]\bigg\}.\nonumber
\end{align}
Notice that $n^{-2} Z_{n}=\psi_{n}$ and $L_{n} = \psi_{n}^{MF}$. \\
For our model, the function  $f:[0,1]^{\binom{n}{2}}\rightarrow\mathbb{R}$ is defined as
\begin{equation}
f(x)=\sum_{i=1}^{n}\sum_{j=1}^{n}\alpha_{ij}x_{ij}+\frac{\beta}{2n}\sum_{i=1}^{n}\sum_{j=1}^{n}\sum_{k=1}^{n}x_{ij}x_{jk}
+\frac{2\gamma}{3n}\sum_{i=1}^{n}\sum_{j=1}^{n}\sum_{k=1}^{n}x_{ij}x_{jk}x_{ki}.
\end{equation}
Then, we can compute that, 
\begin{align}
a & = \Vert f\Vert\leq\sum_{i=1}^{n}\sum_{j=1}^{n}|\alpha_{ij}|
+\frac{1}{2}|\beta|n^{2}+\frac{2}{3}|\gamma|n^{2}\\
 & \leq n^{2}\left[\max_{i,j}|\alpha_{i,j}|
 +\frac{1}{2}|\beta|+\frac{2}{3}|\gamma|\right].\nonumber
\end{align}

Let $k\in\mathbb{N}$, and $H$ be a finite simple graph on the vertex
set $[k]:=\{1,\ldots,k\}$. Let $E$ be the set of edges of $H$ and
$|E|$ be its cardinality. For a function $T:[0,1]^{\binom{n}{2}} \rightarrow\mathbb{R}$
\begin{equation}
T(x):=\frac{1}{n^{k-2}}\sum_{q\in[n]^{k}}\prod_{\{\ell,\ell'\}\in E}x_{q_{\ell}q_{\ell'}},
\end{equation}
\citet{ChatterjeeDembo2014} (Lemma 5.1.) showed that, for any $i<j$, $i'<j'$,
\begin{equation}
\bigg\Vert\frac{\partial T}{\partial x_{ij}}\bigg\Vert\leq2|E|,\label{PartialI}
\end{equation}
and
\begin{equation}
\bigg\Vert\frac{\partial^{2}T}{\partial x_{ij}\partial x_{i'j'}}\bigg\Vert\leq\begin{cases}
4|E|(|E|-1)n^{-1} & \text{if \ensuremath{|\{i,j,i',j'\}|=2} or \ensuremath{3}},\\
4|E|(|E|-1)n^{-2} & \text{if \ensuremath{|\{i,j,i',j'\}|=4}}.
\end{cases}\label{PartialII}
\end{equation}

Therefore, by \eqref{PartialI}, we can compute that
\begin{equation}
b_{(ij)}=\bigg\Vert\frac{\partial f}{\partial x_{ij}}\bigg\Vert\leq 2\max_{i,j}|\alpha_{ij}|+2|\beta|+8|\gamma|.
\end{equation}
By \eqref{PartialII}, we can also compute that
\begin{align}
c_{(i,j)(i'j')} & =\bigg\Vert\frac{\partial^{2}f}{\partial x_{ij}\partial x_{i'j'}}\bigg\Vert\\
 & \leq\begin{cases}
4\left(\frac{1}{2}|\beta|2(2-1)
+\frac{2}{3}|\gamma|3(3-1)\right)n^{-1} & \text{if \ensuremath{|\{i,j,i',j'\}|=2} or \ensuremath{3}},\\
4\left(\frac{1}{2}|\beta|2(2-1)
+\frac{2}{3}|\gamma|3(3-1)\right)n^{-2} & \text{if \ensuremath{|\{i,j,i',j'\}|=4}},
\end{cases}\nonumber
\\
&=\begin{cases}
4\left(|\beta|+4|\gamma|\right)n^{-1} & \text{if \ensuremath{|\{i,j,i',j'\}|=2} or \ensuremath{3}},\\
4\left(|\beta|+4|\gamma|\right)n^{-2} & \text{if \ensuremath{|\{i,j,i',j'\}|=4}}.
\end{cases}\nonumber
\end{align}

Next, we compute that
\begin{eqnarray}
\frac{\partial f}{\partial x_{ij}}=2\alpha_{ij}+\frac{\partial}{\partial x_{ij}}
\left[\frac{\beta}{2n}\sum_{i=1}^{n}\sum_{j=1}^{n}\sum_{k=1}^{n}x_{ij}x_{jk}
+\frac{2\gamma}{3n}\sum_{i=1}^{n}\sum_{j=1}^{n}\sum_{k=1}^{n}x_{ij}x_{jk}x_{ki}\right].
\end{eqnarray}

Let $T_{1}$ and $T_{2}$ be defined as
\begin{equation}
T_{1}(x):=\frac{1}{n}\sum_{i=1}^{n}\sum_{j=1}^{n}\sum_{k=1}^{n}x_{ij}x_{jk},
\qquad
T_{2}(x):=\frac{1}{n}\sum_{i=1}^{n}\sum_{j=1}^{n}\sum_{k=1}^{n}x_{ij}x_{jk}x_{ki}.
\end{equation}
Then, we have
\begin{equation}
\frac{\partial f}{\partial x_{ij}}=2\alpha_{ij}+\frac{\beta}{2}\frac{\partial T_{1}}{\partial x_{ij}}
+\frac{2\gamma}{3}\frac{\partial T_{2}}{\partial x_{ij}}.
\end{equation}
\citet{ChatterjeeDembo2014} (Lemma 5.2.) showed that for the $T_{1}$ and $T_{2}$ defined
above, there exist a set $\mathcal{D}_{1}(\epsilon)$
and $\mathcal{D}_{2}(\epsilon)$ satisfying
the criterion \eqref{fcriterion} (with $f=T_{1}$ and $f=T_{2}$) so that
\begin{align}
&|\mathcal{D}_{1}(\epsilon)|\leq\exp\left\{ \frac{\tilde{C}_{1}2^{4}3^{4}n}{\epsilon^{4}}\log\frac{\tilde{C}_{2}2^{4}3^{4}}{\epsilon^{4}}\right\} =\exp\left\{ \frac{\tilde{C}_{1}6^{4}n}{\epsilon^{4}}
\log\frac{\tilde{C}_{2}6^{4}}{\epsilon^{4}}\right\},
\\
&|\mathcal{D}_{2}(\epsilon)|\leq\exp\left\{ \frac{\tilde{C}_{1}3^{4}3^{4}n}{\epsilon^{4}}\log\frac{\tilde{C}_{2}3^{4}3^{4}}{\epsilon^{4}}\right\}
=\exp\left\{ \frac{\tilde{C}_{1}3^{8}n}{\epsilon^{4}}\log\frac{\tilde{C}_{2}3^{8}}{\epsilon^{4}}\right\}, 
\end{align}
where $\tilde{C}_{1}$ and $\tilde{C}_{2}$ are universal constants.
Let us define
\begin{equation}
\mathcal{D}(\epsilon):=\left\{2\alpha_{ij}+\frac{\beta}{2}d_{1}
+\frac{2\gamma}{3}d_{2}:d_{1}\in\mathcal{D}_{1}\left(\frac{2}{\beta}\cdot\frac{\epsilon}{\sqrt{2}}\right),
d_{2}\in\mathcal{D}_{2}\left(\frac{3}{2\gamma}\cdot\frac{\epsilon}{\sqrt{2}}\right),1\leq i\leq j\leq n\right\} .
\end{equation}
Hence, $\mathcal{D}(\epsilon)$ satisfies the criterion \eqref{fcriterion}
and
\begin{align}
|\mathcal{D}(\epsilon)|
&\leq\frac{1}{2}n(n+1)\left|\mathcal{D}_{1}\left(\sqrt{2}\epsilon/\beta\right)\right|
\cdot\left|\mathcal{D}_{2}\left(3\epsilon/2\sqrt{2}\gamma\right)\right|
\\
&\leq\frac{1}{2}n(n+1)\exp\left\{ \frac{\tilde{C}_{1}6^{4}\beta^{4}n}{4\epsilon^{4}}
\log\frac{\tilde{C}_{2}6^{4}\beta^{4}}{4\epsilon^{4}}\right\}
\exp\left\{ \frac{\tilde{C}_{1}3^{8}2^{6}\gamma^{4}n}{ 3^{4}\epsilon^{4}}\log\frac{\tilde{C}_{2}3^{8}2^{6}\gamma^{4}}{3^{4}\epsilon^{4}}\right\}.
\nonumber
\end{align}
Therefore, by recalling $\mathcal{E}_{1}$ from \eqref{eqn:E1}, we get
\begin{align}\label{E1C1}
\mathcal{E}_{1} & =\frac{1}{4}\left(\binom{n}{2}\sum_{1\leq i<j\leq n}b_{(ij)}^{2}\right)^{1/2}\epsilon+3\binom{n}{2}\epsilon+\log|\mathcal{D}(\epsilon)|\\
 & \leq\left[\frac{1}{4}\left(2\max_{i,j}|\alpha_{ij}|+2|\beta|+8|\gamma|\right)+3\right]\binom{n}{2}\epsilon
 \nonumber\\
 &\qquad\qquad
 +\log\left(\frac{1}{2}n(n+1)\right)+\frac{\tilde{C}_{1}6^{4}\beta^{4}n}{4\epsilon^{4}}\log\frac{\tilde{C}_{2}6^{4}\beta^{4}}{4\epsilon^{4}}
 +\frac{\tilde{C}_{1}3^{4}2^{6}\gamma^{4}n}{\epsilon^{4}}\log\frac{\tilde{C}_{2}3^{4}2^{6}\gamma^{4}}{\epsilon^{4}}\nonumber \\
 & \leq C_{1}(\alpha,\beta,\gamma)n^{2}\epsilon+\frac{C_{1}(\alpha,\beta,\gamma)n}{\epsilon^{4}}\log\frac{C_{1}(\alpha,\beta,\gamma)}{\epsilon^{4}}\nonumber \\
 & =C_{1}(\alpha,\beta,\gamma)n^{9/5}(\log n)^{1/5},\nonumber
\end{align}
by choosing $\epsilon=(\frac{\log n}{n})^{1/5}$, where $C_{1}(\alpha,\beta,\gamma)$
is a constant depending only on $\alpha,\beta,\gamma$:
\begin{equation}\label{eqn:C1}
C_{1}(\alpha,\beta,\gamma):=c_{1}\left(\max_{i,j}|\alpha_{ij}|+|\beta|^{4}+|\gamma|^{4}+1\right),
\end{equation}
where $c_{1}>0$ is some universal constant. 
To see why we can choose $C_{1}(\alpha,\beta,\gamma)$ as in 
\eqref{eqn:C1} so that \eqref{E1C1} holds, we first notice that
it follows from \eqref{E1C1} that 
we can choose $C_{1}(\alpha,\beta,\gamma)$ such that
$C_{1}(\alpha,\beta,\gamma)\geq\max\{\tilde{c}_{1}\max_{ij}|\alpha_{ij}|+\tilde{c}_{2}|\beta|+\tilde{c}_{3}|\gamma|+\tilde{c}_{4},\tilde{c}_{5}\beta^{4},\tilde{c}_{6}\gamma^{4}\}$, 
where $\tilde{c}_{1}$, $\tilde{c}_{2}$, $\tilde{c}_{3}$, $\tilde{c}_{4}$, $\tilde{c}_{5}$, $\tilde{c}_{6}>0$ are some universal constants.
Note that $\max\{\tilde{c}_{1}\max_{ij}|\alpha_{ij}|+\tilde{c}_{2}|\beta|+\tilde{c}_{3}|\gamma|+\tilde{c}_{4},\tilde{c}_{5}\beta^{4},\tilde{c}_{6}\gamma^{4}\}
\leq 
\tilde{c}_{1}\max_{ij}|\alpha_{ij}|+\tilde{c}_{2}|\beta|+\tilde{c}_{3}|\gamma|+\tilde{c}_{4}+\tilde{c}_{5}\beta^{4}+\tilde{c}_{6}\gamma^{4}
\leq
c_{1}\left(\max_{i,j}|\alpha_{ij}|+|\beta|^{4}+|\gamma|^{4}+1\right)$ for some universal constant $c_{1}>0$.
Thus, we can take $C_{1}(\alpha,\beta,\gamma)$
as in \eqref{eqn:C1}.

We can also compute from \eqref{eqn:E2} that
\begin{align*}
\mathcal{E}_{2} & =4\bigg(\sum_{1\leq i<j\leq n}(ac_{(ij)(ij)}+b_{(ij)}^{2})\\
 & \qquad+\frac{1}{4}\sum_{1\leq i<j\leq n,1\leq i'<j'\leq n}\left(ac_{(ij)(i'j')}^{2}+b_{(ij)}b_{(i'j')}c_{(ij)(i'j')}+4b_{(ij)}c_{(ij)(i'j')}\right)\bigg)^{1/2}\nonumber \\
 & \qquad+\frac{1}{4}\left(\sum_{1\leq i<j\leq n}b_{(ij)}^{2}\right)^{1/2}\left(\sum_{1\leq i<j\leq n}c_{(ij)(ij)}^{2}\right)^{1/2}+3\sum_{1\leq i<j\leq n}c_{(ij)(ij)}+\log2,\nonumber
 \end{align*}
 so that
 \begin{align*}
\mathcal{E}_{2} & \leq
 4\bigg\{\binom{n}{2}\left(n\left(\max_{i,j}|\alpha_{ij}|+\frac{1}{2}|\beta|+\frac{2}{3}|\gamma|\right)4(|\beta|+4|\gamma|)
 +\left(2\max_{i,j}|\alpha_{ij}|+2|\beta|+8|\gamma|\right)^{2}\right)\nonumber \\
 & \qquad+\frac{1}{4}n^{2}\left[\max_{i,j}|\alpha_{ij}|+\frac{1}{2}|\beta|+\frac{2}{3}|\gamma|\right]\nonumber
 \\
 &\qquad\cdot\left[\binom{n}{2}\binom{n-2}{2}4^{2}(|\beta|+4|\gamma|)^{2}n^{-4}+\left(\binom{n}{2}^{2}-\binom{n}{2}\binom{n-2}{2}\right)4^{2}(|\beta|+4|\gamma|)^{2}n^{-2}\right]\nonumber \\
 & \qquad+\left(2\max_{i,j}|\alpha_{ij}|+2|\beta|+8|\gamma|\right)
 \cdot\left(\max_{i,j}|\alpha_{ij}|+\frac{1}{2}|\beta|+\frac{2}{3}|\gamma|\right)\nonumber \\
 & \qquad\cdot\left[\binom{n}{2}\binom{n-2}{2}4(|\beta|+4|\gamma|)n^{-2}+\left(\binom{n}{2}^{2}-\binom{n}{2}\binom{n-2}{2}\right)4(|\beta|+4|\gamma|)n^{-1}\right]\bigg\}^{1/2}\nonumber \\
 & \qquad+\frac{1}{4}\binom{n}{2}\left(2\max_{i,j}|\alpha_{ij}|+2|\beta|+8|\gamma|\right)
 4(|\beta|+4|\gamma|)n^{-1}+3\binom{n}{2}4(|\beta|+4|\gamma|)n^{-1}+\log2\nonumber \\
 & \leq C_{2}(\alpha,\beta,\gamma)n^{3/2},\nonumber
\end{align*}
where we used the formulas for $a$, $b_{(ij)}$, and $c_{(ij)(i'j')}$ that 
we derived earlier and the combinatorics identities:
\begin{align*}
&\sum_{1\leq i<j\leq n,1\leq i'<j'\leq n,|\{i,j,i',j'\}|=4}1
=\sum_{1\leq i<j\leq n}\sum_{1\leq i'<j'\leq n,|\{i,j,i',j'\}|=4}1=\binom{n}{2}\binom{n-2}{2},
\\
&\sum_{1\leq i<j\leq n,1\leq i'<j'\leq n,|\{i,j,i',j'\}|=\text{$2$ or $3$}}1
=\binom{n}{2}^{2}-\binom{n}{2}\binom{n-2}{2},
\end{align*}
and $C_{2}(\alpha,\beta,\gamma)$ is a constant depending only on $\alpha,\beta,\gamma$ that
can be chosen as:
\begin{equation}
C_{2}(\alpha,\beta,\gamma):=c_{2}\left(\max_{i,j}|\alpha_{ij}|+|\beta|+
|\gamma|+1\right)^{1/2}(1+|\beta|^{2}+|\gamma|^{2})^{1/2},
\end{equation}
where $c_{2}>0$ is some universal constant.

Finally, to get lower bound, notice that
\begin{equation}
\frac{1}{2}\sum_{1\leq i<j\leq n}c_{(ij)(ij)}\leq\frac{1}{2}\binom{n}{2}4(|\beta|+4|\gamma|)n^{-1}\leq C_{3}(\beta,\gamma)n,
\end{equation}
where $C_{3}(\beta,\gamma)$ is a constant depending only on $\beta,\gamma$ 
and we can simply take $C_{3}(\beta,\gamma)=|\beta|+4|\gamma|$.

\subsection{Proof of Proposition \ref{prop:loglik_approx}.}

 We can approximate $\psi_{n}$ by $\psi_{n}^{MF}$
as seen in Theorem \ref{thm:mf_bounds}, and as a result,
we can approximate the log-likelihood as follows.
\begin{equation*}
\ell_{n}(g,\alpha,\beta,\gamma):=\frac{1}{n^{2}}\log(\pi_{n}(g,\alpha,\beta,\gamma))
=T_{n}(g,\alpha,\beta,\gamma)-\psi_{n}(\alpha,\beta,\gamma),
\end{equation*}
by the mean-field log-likelihood:
\begin{equation*}
\ell_{n}^{MF}(g,\alpha,\beta,\gamma):=T_{n}(g,\alpha,\beta,\gamma)-\psi_{n}^{MF}(\alpha,\beta,\gamma),
\end{equation*}
Then the difference between the mean-field likelihood and the ERGM likelihood is bounded uniformly over $g\in\mathcal{G}$, for any $\alpha,\beta,\gamma$:
\begin{equation*}
0\leq \ell_{n}^{MF}(g,\alpha,\beta,\gamma)-\ell_{n}(g,\alpha,\beta,\gamma)
\leq C_{1}(\alpha,\beta,\gamma)n^{-1/5}(\log n)^{1/5}+C_{2}(\alpha,\beta,\gamma)n^{-1/2}.
\end{equation*}
Therefore, for any compact $\Theta$, we have
\begin{align*}
0 & \leq  \sup_{\alpha,\beta,\gamma\in \Theta}\left[\ell_{n}^{MF}(g,\alpha,\beta,\gamma)-\ell_{n}(g,\alpha,\beta,\gamma)\right] \\
&\leq \sup_{\alpha,\beta,\gamma\in \Theta}
\left[C_{1}(\alpha,\beta,\gamma)n^{-1/5}(\log n)^{1/5}+C_{2}(\alpha,\beta,\gamma)n^{-1/2}\right]
\\
&\leq
\sup_{\alpha,\beta,\gamma\in \Theta}
C_{1}(\alpha,\beta,\gamma)n^{-1/5}(\log n)^{1/5}+\sup_{\alpha,\beta,\gamma\in \Theta}C_{2}(\alpha,\beta,\gamma)n^{-1/2}.
\end{align*}
This proves the result.

%%%%%%%%%%%%%%%%%%%%%%%%%%%%
\section{A Bound Between MLE and Mean-Field Estimator}\label{sec:bound}

We use the bounds on the likelihoods to also derive a bound on the distance between the MLE and our mean-field estimator, when the MLE exists and it is well-behaved. Because our bounds may not be sharp, this proves to be quite hard. We therefore, consider a \emph{local} version of this convergence. 
We know that the ERGM likelihood is concave in parameters because it is an exponential family. We also know that the mean-field log-constant is convex in parameters\footnote{$\psi_{n}^{MF}$ is convex in $(\alpha,\beta,\gamma)$ by its
definition in \eqref{eqn:mean-fieldproblem} since the expression
inside the supremum in \eqref{eqn:mean-fieldproblem} is affine
in $(\alpha,\beta,\gamma)$ and supremum over any affine function
is convex.}, 
therefore the approximate log-likelihood is also concave. However, to get a bound on the distance between estimates we need well-behaved objective functions, with enough curvature  at least close to their maximizers. If the objective functions is too flat, the distance between the estimator may be too large in terms of our upper bounds.\footnote{\cite{GeyerThompson1992} mentions similar problems arise for the MCMC-MLE. Indeed, as mentioned above, the MLE may not exist. For example, if the number of triangles is zero in the data, it will be impossible to estimate $\gamma$ and the MCMC-MLE may give an approximation with solution that tends to infinity. } Therefore we assume that the likelihood and its mean-field approximation have enough curvature.

\begin{proposition}\label{prop:bound}
Assume $(\alpha,\beta,\gamma)$ lives on a 
compact set $\Theta$. Let
$\hat{\theta}_{n}:=(\hat{\alpha}_{n},\hat{\beta}_{n},\hat{\gamma}_{n})$
and $\hat{\theta}_{n}^{MF}:=(\hat{\alpha}_{n}^{MF},\hat{\beta}_{n}^{MF},\hat{\gamma}_{n}^{MF})$ be the maximizers of  $\ell_{n}$ and $\ell_{n}^{MF}$, respectively, in the interior of $\Theta$.
Moreover, we assume that $\psi_{n}$ and $\psi_{n}^{MF}$ are differentiable and $\mu_{n}$- and $\mu_n^{MF}$-strongly convex in $(\alpha,\beta,\gamma)$, respectively, on $\Theta$, where $\mu_{n}>0$ and $\mu_{n}^{MF}>0$. Then
\begin{equation}\label{theta:bound}
\Vert\hat{\theta}_{n}-\hat{\theta}_{n}^{MF}\Vert
\leq\frac{2}{(\mu_{n}+\mu_{n}^{MF})^{\frac{1}{2}}}
\left[\sup_{\alpha,\beta,\gamma\in \Theta}
C_{1}^{\frac{1}{2}}(\alpha,\beta,\gamma)\left(\frac{\log n}{n} \right)^{\frac{1}{10}}+\sup_{\alpha,\beta,\gamma\in \Theta}C_{2}^{\frac{1}{2}}(\alpha,\beta,\gamma)n^{-\frac{1}{4}}\right],
\end{equation}
where $C_{1}$ and $C_{2}$ are defined in Theorem \ref{thm:mf_bounds}
and $\Vert\cdot\Vert$ denotes the Euclidean norm.
\end{proposition}
%\begin{proof}See Appendix.\end{proof}

In Proposition \ref{prop:bound}, if $\mu_{n}$ and $\mu_{n}^{MF}$ goes to
zero sufficiently fast as $n$ goes zero, then 
the bound in \eqref{theta:bound} may not go to zero as $n$ goes to zero. 
If for example $\mu_{n},\mu_{n}^{MF}$ are uniformly bounded from 
below, and both $\sup_{\alpha,\beta,\gamma\in \Theta}C_{1}(\alpha,\beta,\gamma)$
and $\sup_{\alpha,\beta,\gamma\in \Theta}C_{2}(\alpha,\beta,\gamma)$
are $O(1)$, then 
$\Vert\hat{\theta}_{n}-\hat{\theta}_{n}^{MF}\Vert=O(n^{-1/10}(\log n)^{1/10})$.

\subsection{Proof of Proposition \ref{prop:bound}}

We assume that 
$\psi_{n}$ (resp. $\psi_{n}^{MF}$)
is differentiable and $\mu_{n}$-strongly convex (resp. $\mu_{n}^{MF}$-strongly convex) in $\theta:=(\alpha,\beta,\gamma)\in\Theta$. 
Note that
\begin{equation*}
\ell_{n}=T_{n}-\psi_{n},
\qquad
\ell_{n}^{MF}=T_{n}-\psi_{n}^{MF},
\end{equation*}
and $T_{n}$ is linear in $\theta=(\alpha,\beta,\gamma)$, 
we have that $\ell_{n}$ (resp. $\ell_{n}^{MF}$)
is differentiable and $\mu_{n}$-strongly concave in $\theta:=(\alpha,\beta,\gamma)\in\Theta$
so that for any $x,y\in\Theta$, 
\begin{equation}
\ell_{n}(y)\leq\ell_{n}(x)+\nabla\ell_{n}(x)^{T}(y-x)-\frac{\mu_{n}}{2}\Vert y-x\Vert^{2},
\end{equation}
and in particular,
\begin{align}\label{eqn:add:1}
\ell_{n}(\hat{\theta}_{n}^{MF})
&\leq\ell_{n}(\hat{\theta}_{n})
+\nabla\ell_{n}(\hat{\theta}_{n})^{T}(\hat{\theta}_{n}^{MF}-\hat{\theta}_{n})
-\frac{\mu_{n}}{2}\Vert\hat{\theta}_{n}^{MF}-\hat{\theta}_{n}\Vert^{2}
\\
&=\ell_{n}(\hat{\theta}_{n})
-\frac{\mu_{n}}{2}\Vert\hat{\theta}_{n}^{MF}-\hat{\theta}_{n}\Vert^{2},
\nonumber
\end{align}
and similarly, for any $x,y\in\Theta$,
\begin{equation}
\ell_{n}^{MF}(y)\leq\ell_{n}^{MF}(x)+\nabla \ell_{n}^{MF}(x)^{T}(y-x)-\frac{\mu_{n}}{2}\Vert y-x\Vert^{2},
\end{equation}
and in particular,
\begin{align}\label{eqn:add:2}
\ell_{n}^{MF}(\hat{\theta}_{n})
&\leq\ell_{n}^{MF}(\hat{\theta}_{n}^{MF})
+\nabla\ell_{n}^{MF}(\hat{\theta}_{n}^{MF})^{T}(\hat{\theta}_{n}-\hat{\theta}_{n}^{MF})
-\frac{\mu_{n}^{MF}}{2}\Vert\hat{\theta}_{n}-\hat{\theta}_{n}^{MF}\Vert^{2}
\\
&=\ell_{n}^{MF}(\hat{\theta}_{n}^{MF})
-\frac{\mu_{n}^{MF}}{2}\Vert\hat{\theta}_{n}-\hat{\theta}_{n}^{MF}\Vert^{2}.
\nonumber
\end{align}
Adding the inequalities \eqref{eqn:add:1} and \eqref{eqn:add:2}, we get
\begin{align*}
\Vert\hat{\theta}_{n}-\hat{\theta}_{n}^{MF}\Vert^{2}
&\leq\frac{2}{\mu_{n}^{MF}+\mu_{n}}
\left[\left(\ell_{n}^{MF}(\hat{\theta}_{n}^{MF})-\ell_{n}(\hat{\theta}_{n}^{MF})\right)
+\left(\ell_{n}(\hat{\theta}_{n})-\ell_{n}^{MF}(\hat{\theta}_{n})\right)\right]
\\
&\leq
\frac{4}{\mu_{n}^{MF}+\mu_{n}}\sup_{\theta\in\Theta}|\ell_{n}^{MF}(\theta)-\ell_{n}(\theta)|.
\end{align*}
By applying Theorem \ref{thm:mf_bounds}, we get
\begin{align*}
\Vert\hat{\theta}_{n}-\hat{\theta}_{n}^{MF}\Vert
&\leq\frac{2}{(\mu_{n}+\mu_{n}^{MF})^{\frac{1}{2}}}
\left[\sup_{\alpha,\beta,\gamma\in \Theta}
C_{1}(\alpha,\beta,\gamma)n^{-\frac{1}{5}}(\log n)^{\frac{1}{5}}+\sup_{\alpha,\beta,\gamma\in \Theta}C_{2}(\alpha,\beta,\gamma)n^{-\frac{1}{2}}\right]^{\frac{1}{2}}
\\
&\leq\frac{2}{(\mu_{n}+\mu_{n}^{MF})^{\frac{1}{2}}}
\left[\sup_{\alpha,\beta,\gamma\in \Theta}
C_{1}^{\frac{1}{2}}(\alpha,\beta,\gamma)n^{-\frac{1}{10}}(\log n)^{\frac{1}{10}}+\sup_{\alpha,\beta,\gamma\in \Theta}C_{2}^{\frac{1}{2}}(\alpha,\beta,\gamma)n^{-\frac{1}{4}}\right],
\end{align*}
where the last step is due to the inequality
$\sqrt{x+y}\leq\sqrt{x}+\sqrt{y}$ for any $x,y\geq 0$.
The proof is complete.

%%%%%%%%%%%%%%%%%%%%%%%%%%%%%%%%%%
\section{Additional simulation results}

\subsection{No covariates, edges and two-stars model}
We have estimated a model with no covariates. This corresponds to a model in which $\tilde{\alpha}_2=0$ or $\alpha_1=\alpha_2=\alpha$. The results of our simulations for small networks are in Table \ref{tab:MC_-2_0_1}.  Our method performs relatively well in this simpler case. Indeed in this case there are results that would allow us to solve the variational problem in closed form for large $n$ \citep{DiaconisChatterjee2011, Mele2010a, AristoffZhu2014, RadinYin2013}. The MPLE and MCMC-MLE median estimate seems to converge to the true value as we increase $n$, but our approximation seems to perform slightly better here.

\begin{table}[!ht]
\caption{Monte Carlo estimates, comparison of three methods. True parameter vector is $(\tilde{\alpha}_1, \tilde{\alpha}_2, \beta)=(-2,0,1)$}
\centering
\begin{tabular}{lccccccccc}
  \hline
$n=50$ & \multicolumn{3}{c}{MCMC-MLE} & \multicolumn{3}{c}{MEAN-FIELD} & \multicolumn{3}{c}{MPLE} \\
  \hline
 & $\tilde{\alpha}_1$ & $\tilde{\alpha}_2$ & $\beta$ & $\tilde{\alpha}_1$ & $\tilde{\alpha}_2$ & $\beta$ & $\tilde{\alpha}_1$ & $\tilde{\alpha}_2$ & $\beta$  \\
  \hline
  median & -2.063 & 0.016 & -0.324 & -2.021 & 0.007 & 0.999 & -1.983 & 0.018 & -1.006 \\ 
  0.05 & -2.692 & -0.614 & -23.828 & -2.412 & -0.372 & 0.975 & -2.439 & -0.368 & -34.177 \\ 
  0.95 & -1.363 & 0.657 & 22.738 & -1.783 & 0.413 & 1.015 & -1.449 & 0.401 & 14.465 \\ 
   \hline\hline

$n=100$ & \multicolumn{3}{c}{MCMC-MLE} & \multicolumn{3}{c}{MEAN-FIELD} & \multicolumn{3}{c}{MPLE} \\
  \hline
 & $\tilde{\alpha}_1$ & $\tilde{\alpha}_2$ & $\beta$ & $\tilde{\alpha}_1$ & $\tilde{\alpha}_2$ & $\beta$ & $\tilde{\alpha}_1$ & $\tilde{\alpha}_2$ & $\beta$  \\
  \hline
  median & -1.970 & -0.042 & 0.221 & -1.981 & -0.017 & 1.000 & -1.949 & -0.023 & -1.231 \\ 
  0.05 & -2.241 & -0.333 & -13.226 & -2.101 & -0.194 & 0.993 & -2.168 & -0.196 & -14.402 \\ 
  0.95 & -1.602 & 0.249 & 16.316 & -1.874 & 0.134 & 1.012 & -1.643 & 0.142 & 9.328 \\ 
    \hline\hline

$n=200$ & \multicolumn{3}{c}{MCMC-MLE} & \multicolumn{3}{c}{MEAN-FIELD} & \multicolumn{3}{c}{MPLE} \\
  \hline
 & $\tilde{\alpha}_1$ & $\tilde{\alpha}_2$ & $\beta$ & $\tilde{\alpha}_1$ & $\tilde{\alpha}_2$ & $\beta$ & $\tilde{\alpha}_1$ & $\tilde{\alpha}_2$ & $\beta$  \\
  \hline
  median & -2.012 & -0.005 & 1.483 & -1.998 & 0.002 & 1.000 & -2.003 & -0.001 & 1.225 \\ 
  0.05 & -2.214 & -0.184 & -9.515 & -2.067 & -0.093 & 0.997 & -2.160 & -0.095 & -9.682 \\ 
  0.95 & -1.796 & 0.161 & 12.179 & -1.935 & 0.091 & 1.003 & -1.790 & 0.095 & 8.784 \\ 
   \hline\hline
\end{tabular}
\flushleft Notes. See notes for Table \ref{tab:MC_-2_1_1_1_twostar_triangles}.
\label{tab:MC_-2_0_1}
\end{table}

\subsection{Model with 2-stars}
In this subsection we report estimates of a model where the triangle term is excluded from the specification ( $\gamma=0$ in log-likelihood \eqref{eq:loglik_simulations_R}). In Table \ref{tab:MC_-2_1_2} we report results for 100 simulations of a model with $(\tilde{\alpha}_1, \tilde{\alpha}_2, \beta)=(-2,1,2)$. We run simulations for networks of size $n=50,100,200$, to show how our method compares to MCMC-MLE and MPLE when the size of the network grows. In general, we expect more precise results as $n$ grows large. 

The results are encouraging and the mean-field approximation seems to behave as expected. Indeed, the median estimate is very close to the true parameters that generate the data. As the size of the network grows from $n=50$ to $n=200$, both MCMC-MLE and MPLE also improve in precision. The fastest method in terms of computational time is the MPLE. This is because the MPLE's speed depends on the number of parameters. Our mean-field approximation is as fast as the MCMC-MLE.

\begin{table}[!ht]
\caption{Monte Carlo estimates, comparison of three methods. True parameter vector is $(\tilde{\alpha}_1, \tilde{\alpha}_2, \beta)=(-2,1,2)$}
\centering
\begin{tabular}{lccccccccc}
  \hline\hline
$n=50$ & \multicolumn{3}{c}{MCMC-MLE} & \multicolumn{3}{c}{MEAN-FIELD} & \multicolumn{3}{c}{MPLE} \\
  \hline
 & $\tilde{\alpha}_1$ & $\tilde{\alpha}_2$ & $\beta$ & $\tilde{\alpha}_1$ & $\tilde{\alpha}_2$ & $\beta$ & $\tilde{\alpha}_1$ & $\tilde{\alpha}_2$ & $\beta$  \\
  \hline
% latex table generated in R 3.5.2 by xtable 1.8-3 package
% Wed Mar 27 14:11:51 2019
  median & -2.015 & 0.999 & 2.303 & -1.993 & 1.000 & 2.004 & -1.996 & 0.998 & 1.820 \\ 
  0.05 & -2.433 & 0.641 & -1.085 & -2.060 & 0.885 & 1.916 & -2.325 & 0.780 & -2.556 \\ 
  0.95 & -1.666 & 1.337 & 6.118 & -1.905 & 1.090 & 2.087 & -1.573 & 1.273 & 4.783 \medskip\\ 
   \hline\hline
$n=100$ & \multicolumn{3}{c}{MCMC-MLE} & \multicolumn{3}{c}{MEAN-FIELD} & \multicolumn{3}{c}{MPLE} \\
  \hline
 & $\tilde{\alpha}_1$ & $\tilde{\alpha}_2$ & $\beta$ & $\tilde{\alpha}_1$ & $\tilde{\alpha}_2$ & $\beta$ & $\tilde{\alpha}_1$ & $\tilde{\alpha}_2$ & $\beta$  \\
  \hline
  median & -1.995 & 1.012 & 1.932 & -1.980 & 1.011 & 2.011 & -1.980 & 1.010 & 1.783 \\ 
  0.05 & -2.189 & 0.861 & 0.701 & -2.032 & 0.969 & 1.992 & -2.175 & 0.901 & 0.329 \\ 
  0.95 & -1.833 & 1.157 & 3.314 & -1.944 & 1.044 & 2.088 & -1.816 & 1.141 & 2.867  \medskip\\
   \hline
$n=200$ & \multicolumn{3}{c}{MCMC-MLE} & \multicolumn{3}{c}{MEAN-FIELD} & \multicolumn{3}{c}{MPLE} \\
  \hline
 & $\tilde{\alpha}_1$ & $\tilde{\alpha}_2$ & $\beta$ & $\tilde{\alpha}_1$ & $\tilde{\alpha}_2$ & $\beta$ & $\tilde{\alpha}_1$ & $\tilde{\alpha}_2$ & $\beta$  \\
  \hline
  median & -2.000 & 1.009 & 1.938 & -1.986 & 1.005 & 2.016 & -1.997 & 1.007 & 1.930 \\ 
  0.05 & -2.182 & 0.925 & 0.843 & -2.004 & 0.932 & 1.999 & -2.176 & 0.950 & 0.592 \\ 
  0.95 & -1.882 & 1.087 & 4.119 & -1.935 & 1.028 & 2.214 & -1.847 & 1.069 & 3.541 \medskip\\ 
   \hline\hline
\end{tabular}
\flushleft Notes. See notes for Table \ref{tab:MC_-2_1_1_1_twostar_triangles}.
\label{tab:MC_-2_1_2}
\end{table}

The second set of Monte Carlo experiments is reported in Table \ref{tab:MC_-2_1_3}, where the data are generated by parameter vector $(\tilde{\alpha}_1, \tilde{\alpha}_2, \beta)=(-2,1,3)$. The pattern is similar to the previous table, but the mean field estimates exhibit a little more bias.

\begin{table}[!ht]
\caption{Monte Carlo estimates, comparison of three methods. True parameter vector is $(\tilde{\alpha}_1, \tilde{\alpha}_2, \beta)=(-2,1,3)$}
\centering
\begin{tabular}{lccccccccc}
  \hline
$n=50$ & \multicolumn{3}{c}{MCMC-MLE} & \multicolumn{3}{c}{MEAN-FIELD} & \multicolumn{3}{c}{MPLE} \\
  \hline
 & $\tilde{\alpha}_1$ & $\tilde{\alpha}_2$ & $\beta$ & $\tilde{\alpha}_1$ & $\tilde{\alpha}_2$ & $\beta$ & $\tilde{\alpha}_1$ & $\tilde{\alpha}_2$ & $\beta$  \\
  \hline
  median & -1.978 & 1.010 & 2.742 & -1.958 & 1.026 & 3.025 & -1.921 & 1.016 & 2.357 \\ 
  0.05 & -2.308 & 0.745 & 1.342 & -2.045 & 0.878 & 2.938 & -2.201 & 0.823 & -0.742 \\ 
  0.95 & -1.689 & 1.229 & 4.466 & -1.811 & 1.141 & 3.468 & -1.547 & 1.202 & 4.288  \medskip\\
   \hline\hline

$n=100$ & \multicolumn{3}{c}{MCMC-MLE} & \multicolumn{3}{c}{MEAN-FIELD} & \multicolumn{3}{c}{MPLE} \\
  \hline
 & $\tilde{\alpha}_1$ & $\tilde{\alpha}_2$ & $\beta$ & $\tilde{\alpha}_1$ & $\tilde{\alpha}_2$ & $\beta$ & $\tilde{\alpha}_1$ & $\tilde{\alpha}_2$ & $\beta$  \\
  \hline
  median & -2.005 & 1.002 & 3.022 & -1.851 & 1.091 & 3.166 & -1.997 & 1.001 & 3.009 \\ 
  0.05 & -2.116 & 0.892 & 2.665 & -2.274 & 0.866 & 2.998 & -2.098 & 0.924 & 2.514 \\ 
  0.95 & -1.902 & 1.110 & 3.414 & -1.670 & 1.861 & 4.092 & -1.895 & 1.096 & 3.425  \medskip\\
   \hline\hline

$n=200$ & \multicolumn{3}{c}{MCMC-MLE} & \multicolumn{3}{c}{MEAN-FIELD} & \multicolumn{3}{c}{MPLE} \\
  \hline
 & $\tilde{\alpha}_1$ & $\tilde{\alpha}_2$ & $\beta$ & $\tilde{\alpha}_1$ & $\tilde{\alpha}_2$ & $\beta$ & $\tilde{\alpha}_1$ & $\tilde{\alpha}_2$ & $\beta$  \\
  \hline
  median & -2.003 & 1.000 & 2.959 & -1.923 & 1.030 & 3.107 & -1.984 & 1.000 & 2.847 \\ 
  0.05 & -2.151 & 0.934 & 2.314 & -2.059 & 0.922 & 3.000 & -2.104 & 0.951 & 2.096 \\ 
  0.95 & -1.902 & 1.064 & 3.944 & -1.836 & 1.164 & 4.222 & -1.861 & 1.039 & 3.666 \\ 
   \hline\hline
\end{tabular}
\flushleft Notes. See notes for Table \ref{tab:MC_-2_1_1_1_twostar_triangles}.
\label{tab:MC_-2_1_3}
\end{table}

\subsection{Model with triangles}
The second set of simulations involves a model with no two-stars, that is $\beta=0$, in Table \ref{tab:MC_-2_1_-2_triangles}. In this specification our mean-field approximation seems to do better than the other estimators, at least for this small networks.

\begin{table}[!ht]
\caption{Monte Carlo estimates, comparison of three methods. True parameter vector is $(\tilde{\alpha}_1, \tilde{\alpha}_2, \gamma)=(-2,1,-2)$}
\centering
\begin{tabular}{lccccccccc}
  \hline
$n=50$ & \multicolumn{3}{c}{MCMC-MLE} & \multicolumn{3}{c}{MEAN-FIELD} & \multicolumn{3}{c}{MPLE} \\
  \hline
 & $\tilde{\alpha}_1$ & $\tilde{\alpha}_2$ & $\gamma  $ & $\tilde{\alpha}_1$ & $\tilde{\alpha}_2$ & $\gamma$ & $\tilde{\alpha}_1$ & $\tilde{\alpha}_2$ & $\gamma $  \\
  \hline
  median & -2.024 & 1.026 & -13.959 & -2.000 & 1.005 & -2.000 & -2.031 & 1.012 & -9.804 \\ 
  0.05 & -2.384 & 0.622 & -60.419 & -2.321 & 0.168 & -6.425 & -2.398 & 0.758 & -45.881 \\ 
  0.95 & -1.689 & 1.457 & 49.585 & -0.739 & 2.246 & -1.777 & -1.809 & 1.394 & 21.696  \medskip\\
   \hline\hline

$n=100$ & \multicolumn{3}{c}{MCMC-MLE} & \multicolumn{3}{c}{MEAN-FIELD} & \multicolumn{3}{c}{MPLE} \\
  \hline
 & $\tilde{\alpha}_1$ & $\tilde{\alpha}_2$ & $\gamma  $ & $\tilde{\alpha}_1$ & $\tilde{\alpha}_2$ & $\gamma$ & $\tilde{\alpha}_1$ & $\tilde{\alpha}_2$ & $\gamma $  \\
  \hline
   median & -2.006 & 1.019 & -6.053 & -1.967 & 1.035 & -2.007 & -2.002 & 1.015 & -4.980 \\ 
  0.05 & -2.164 & 0.832 & -35.171 & -3.472 & 0.951 & -7.368 & -2.124 & 0.876 & -23.937 \\ 
  0.95 & -1.824 & 1.183 & 27.361 & -1.388 & 3.763 & -1.910 & -1.890 & 1.153 & 13.519   \medskip\\
   \hline\hline

$n=200$ & \multicolumn{3}{c}{MCMC-MLE} & \multicolumn{3}{c}{MEAN-FIELD} & \multicolumn{3}{c}{MPLE} \\
  \hline
 & $\tilde{\alpha}_1$ & $\tilde{\alpha}_2$ & $\gamma  $ & $\tilde{\alpha}_1$ & $\tilde{\alpha}_2$ & $\gamma$ & $\tilde{\alpha}_1$ & $\tilde{\alpha}_2$ & $\gamma $  \\
  \hline
  median & -2.007 & 1.001 & -1.002 & -1.972 & 1.031 & -2.006 & -2.003 & 1.000 & -1.913 \\ 
  0.05 & -2.083 & 0.901 & -23.049 & -2.014 & 1.008 & -2.115 & -2.061 & 0.929 & -15.721 \\ 
  0.95 & -1.931 & 1.095 & 16.760 & -1.473 & 1.636 & -1.983 & -1.952 & 1.072 & 9.153 \medskip\\ 
   \hline\hline
\end{tabular}
\flushleft Notes. See notes for Table \ref{tab:MC_-2_1_1_1_twostar_triangles}.
\label{tab:MC_-2_1_-2_triangles}
\end{table}

\newpage

\begin{table}[ht]
\caption{Monte Carlo estimates, comparison of three methods. True parameter vector is $(\tilde{\alpha}_1, \tilde{\alpha}_2, \beta,\gamma)=(-2,1,-1,-1)$}
\centering
\begin{small}
\begin{tabular}{lcccccccccccc}
  \hline
$n=50$ & \multicolumn{4}{c}{MCMC-MLE} & \multicolumn{4}{c}{MEAN-FIELD} & \multicolumn{4}{c}{MPLE} \\
  \hline
 & $\tilde{\alpha}_1$ & $\tilde{\alpha}_2$ & $\beta$ & $\gamma  $ & $\tilde{\alpha}_1$ & $\tilde{\alpha}_2$ & $\beta$  & $\gamma$ & $\tilde{\alpha}_1$ & $\tilde{\alpha}_2$ & $\beta$  & $\gamma $  \\
  \hline
median & -2.008 & 1.023 & -1.256 & -4.943 & -1.977 & 1.030 & -1.018 & -1.002 & -1.959 & 1.015 & -2.032 & -3.296 \\ 
  mad & 0.320 & 0.267 & 4.898 & 43.074 & 0.153 & 0.165 & 0.144 & 0.154 & 0.307 & 0.191 & 4.532 & 24.826 \\ 
 \hline\hline

$n=100$ & \multicolumn{4}{c}{MCMC-MLE} & \multicolumn{4}{c}{MEAN-FIELD} & \multicolumn{4}{c}{MPLE} \\
  \hline
median & -1.996 & 1.004 & -1.138 & -3.173 & -1.932 & 1.177 & -1.057 & -1.021 & -1.974 & 1.006 & -1.566 & -1.489 \\ 
  mad & 0.219 & 0.133 & 3.364 & 28.410 & 0.567 & 0.553 & 0.335 & 0.346 & 0.207 & 0.093 & 3.119 & 16.695 \\ 
  \hline\hline

$n=200$& \multicolumn{4}{c}{MCMC-MLE} & \multicolumn{4}{c}{MEAN-FIELD} & \multicolumn{4}{c}{MPLE} \\
  \hline
median & -1.995 & 1.007 & -1.155 & -0.980 & -1.603 & 1.645 & -1.317 & -1.078 & -1.987 & 1.003 & -1.340 & -1.308 \\ 
  mad & 0.133 & 0.069 & 2.098 & 18.167 & 0.559 & 0.794 & 0.656 & 0.558 & 0.127 & 0.047 & 2.064 & 11.196 \\ 
   \hline\hline

$n=500$& \multicolumn{4}{c}{MCMC-MLE} & \multicolumn{4}{c}{MEAN-FIELD} & \multicolumn{4}{c}{MPLE} \\
  \hline
median & -1.998 & 1.002 & -1.070 & -1.315 & -1.682 & 1.836 & -1.431 & -1.155 & -1.991 & 1.000 & -1.113 & -1.227 \\ 
  mad & 0.084 & 0.033 & 1.496 & 10.897 & 0.805 & 0.849 & 0.776 & 0.883 & 0.079 & 0.020 & 1.340 & 7.036 \\ 
   \hline\hline

%$n=1000$& \multicolumn{4}{c}{MCMC-MLE} & \multicolumn{4}{c}{MEAN-FIELD} & \multicolumn{4}{c}{MPLE} \\
%  \hline
% & $\tilde{\alpha}_1$ & $\tilde{\alpha}_2$ & $\beta$ & $\gamma  $ & $\tilde{\alpha}_1$ & $\tilde{\alpha}_2$ & $\beta$  & %$\gamma$ & $\tilde{\alpha}_1$ & $\tilde{\alpha}_2$ & $\beta$  & $\gamma $  \\
%   \hline
%*  median & -1.976 & 0.986 & 1.578 & 0.483 & -1.844 & 0.916 & 5.223 & -1.000 & -1.970 & 0.994 & 1.551 & -0.960 \\ 
%  0.05 & -2.191 & 0.894 & -0.231 & -13.421 & -1.978 & 0.812 & 4.321 & -1.084 & -2.173 & 0.924 & -0.822 & -11.144 \\ 
%  0.95 & -1.805 & 1.089 & 3.920 & 14.010 & -1.783 & 1.033 & 6.900 & -0.920 & -1.736 & 1.076 & 3.857 & 11.229 \medskip\\ 
%   \hline\hline
\end{tabular}
\end{small}
\flushleft Notes: see notes for Table \ref{tab:MC_-2_1_1_1_twostar_triangles}.
\label{tab:MC_-2_-1_-1_-1_twostar_triangles}
\end{table}

\begin{table}[ht]
\caption{Monte Carlo estimates, comparison of three methods. True parameter vector is $(\tilde{\alpha}_1, \tilde{\alpha}_2, \beta,\gamma)=(-2,1,-2,3)$}
\centering
\begin{small}
\begin{tabular}{lcccccccccccc}
  \hline
$n=50$ & \multicolumn{4}{c}{MCMC-MLE} & \multicolumn{4}{c}{MEAN-FIELD} & \multicolumn{4}{c}{MPLE} \\
  \hline
 & $\tilde{\alpha}_1$ & $\tilde{\alpha}_2$ & $\beta$ & $\gamma  $ & $\tilde{\alpha}_1$ & $\tilde{\alpha}_2$ & $\beta$  & $\gamma$ & $\tilde{\alpha}_1$ & $\tilde{\alpha}_2$ & $\beta$  & $\gamma $  \\
  \hline
median & -2.005 & 1.024 & -2.368 & -4.197 & -1.955 & 1.037 & -2.022 & 2.998 & -1.958 & 1.017 & -3.006 & -0.198 \\ 
  mad & 0.349 & 0.292 & 5.767 & 46.688 & 0.095 & 0.085 & 0.088 & 0.082 & 0.307 & 0.196 & 4.707 & 26.076 \\ 
 \hline\hline

$n=100$ & \multicolumn{4}{c}{MCMC-MLE} & \multicolumn{4}{c}{MEAN-FIELD} & \multicolumn{4}{c}{MPLE} \\
  \hline
median & -2.000 & 0.995 & -2.333 & 1.560 & -1.909 & 1.082 & -2.100 & 2.983 & -1.972 & 0.997 & -2.708 & 2.617 \\ 
  mad & 0.216 & 0.145 & 3.429 & 31.810 & 0.151 & 0.147 & 0.199 & 0.130 & 0.195 & 0.099 & 3.221 & 17.184 \\ 
  \hline\hline

$n=200$& \multicolumn{4}{c}{MCMC-MLE} & \multicolumn{4}{c}{MEAN-FIELD} & \multicolumn{4}{c}{MPLE} \\
  \hline
median & -1.998 & 0.997 & -2.062 & 1.847 & -1.593 & 1.512 & -2.849 & 2.711 & -1.985 & 0.999 & -2.321 & 2.326 \\ 
  mad & 0.129 & 0.073 & 2.302 & 22.032 & 0.565 & 0.677 & 1.195 & 0.594 & 0.124 & 0.049 & 2.167 & 13.057 \\ 
   \hline\hline

$n=500$& \multicolumn{4}{c}{MCMC-MLE} & \multicolumn{4}{c}{MEAN-FIELD} & \multicolumn{4}{c}{MPLE} \\
  \hline
median & -2.004 & 1.002 & -1.944 & 2.531 & -1.523 & 1.605 & -3.493 & 2.557 & -2.002 & 1.002 & -2.059 & 2.786 \\ 
  mad & 0.091 & 0.038 & 1.579 & 11.813 & 0.782 & 0.726 & 1.472 & 0.982 & 0.080 & 0.024 & 1.472 & 8.068 \\ 
   \hline\hline

%$n=1000$& \multicolumn{4}{c}{MCMC-MLE} & \multicolumn{4}{c}{MEAN-FIELD} & \multicolumn{4}{c}{MPLE} \\
%  \hline
% & $\tilde{\alpha}_1$ & $\tilde{\alpha}_2$ & $\beta$ & $\gamma  $ & $\tilde{\alpha}_1$ & $\tilde{\alpha}_2$ & $\beta$  & %$\gamma$ & $\tilde{\alpha}_1$ & $\tilde{\alpha}_2$ & $\beta$  & $\gamma $  \\
%   \hline
%*  median & -1.976 & 0.986 & 1.578 & 0.483 & -1.844 & 0.916 & 5.223 & -1.000 & -1.970 & 0.994 & 1.551 & -0.960 \\ 
%  0.05 & -2.191 & 0.894 & -0.231 & -13.421 & -1.978 & 0.812 & 4.321 & -1.084 & -2.173 & 0.924 & -0.822 & -11.144 \\ 
%  0.95 & -1.805 & 1.089 & 3.920 & 14.010 & -1.783 & 1.033 & 6.900 & -0.920 & -1.736 & 1.076 & 3.857 & 11.229 \medskip\\ 
%   \hline\hline
\end{tabular}
\end{small}
\flushleft Notes: see notes for Table \ref{tab:MC_-2_1_1_1_twostar_triangles}.
\label{tab:MC_-2_1_-2_3_twostar_triangles}
\end{table}

\subsection{Some examples of nonconvergence}
In Tables \ref{tab:MC_-2_-1_-1_-1_twostar_triangles} and \ref{tab:MC_-2_1_-2_3_twostar_triangles} we show examples in which our mean-field approximation performs worse than the alternative estimators. 
There are several possible explanations for this poor convergence. First, it may be that we are not finding the maximizer of the approximation variational problem \eqref{eqn:mean-fieldproblem}, given the local nature of updates \eqref{eq:mean-field_update}. In these simulations we do not start the matrix $\bm{\mu}^{(0)}$ at different initial values, therefore we converge to a local maximum that may not be global. Our package \texttt{mfergm} allows the researcher to initialize $\bm{\mu}^{(0)}$ at different random starting points. This can improve convergence. In principle we should increase the number of re-starts as $n$ grows, as it is known that these models may have multiple modes. Ideally, one can use a Nelder-Mead or Simulated Annealing algorithm to find the maximizer of the variational problem, but this is more time-consuming. All these ideas lead to simple parallelization of our package's functions that are beyond the scope of the present work. 
Second, the tolerance level that we use $\epsilon_{tol}=0.0001$ may be too large. Third, the likelihood may exhibit a phase transition and thus a small difference in parameters may cause a large change in the behavior of the model. We conjecture that some of these issues are related to identification and we plan to explore this in future work.

\subsection{A note on computational speed}
In our Monte Carlo exercises, we note that the computational speed of the three estimators is similar for small networks. For $n=100$, the mean-field approximation takes about $3.5s$ to estimate the model, while an MCMC-MLE with a burnin of $100,000$ and sampling every $1000$ iterations takes approximately $5.5s$ and the MPLE takes about $1.7s$. For $n=50$ the estimates take $1.6s$ for mean-field, $4s$ for MCMC-MLE and $1.2s$ for MPLE. 

However, for larger networks, our code is computationally inefficient and results in much larger computational time than using the built-in functions in the \texttt{ergm} package in R for MCMC-MLE and MPLE. We have experimented with faster iterative routines that could speed up the approximate solution of the variational mean-field problem, but these are not fully stable. Additionally our code does not make efficient use of the memory, as the matrix $\mu$ is dense and we are not using efficient matrix algebra libraries to speed up the computation. We believe that such improvement in our benchmark code will make computational time comparable to MPLE.

\newpage
\newpage

\begin{center} \begin{Large} \textbf{ONLINE APPENDIX - NOT FOR PUBLICATION}\end{Large}\end{center}

\section{Asymptotic Results}

In this section we consider the model as $n\rightarrow\infty$. 
We have seen previously that the log normalizing constant
$\psi_{n}(\alpha,\beta,\gamma)$ can be approximated
by $\psi_{n}^{MF}(\bm{\mu}(\alpha,\beta,\gamma))$
by the mean-field approximation, 
where $\bm{\mu}(\alpha,\beta,\gamma)$ solves
the optimization problem in \eqref{eqn:mean-fieldproblem}
and $\psi_{n}^{MF}(\bm{\mu}(\alpha,\beta,\gamma))$
is its optimal value, where we recall that
\begin{align*}
\psi_{n}^{MF}(\bm{\mu}(\alpha,\beta,\gamma))
&=\sup_{\bm{\mu}\in[0,1]^{n^{2}}:\mu_{ij}=\mu_{ji},\forall i,j}
\Bigg\{\frac{1}{n^{2}}\sum_{i,j}\alpha_{ij}\mu_{ij}
+\frac{\beta}{2n^{3}}\sum_{i,j,k}\mu_{ij}\mu_{jk}
+\frac{2\gamma}{3n^{3}}\sum_{i,j,k}\mu_{ij}\mu_{jk}\mu_{ki}
\\
&\qquad\qquad
-\frac{1}{2n^{2}}\sum_{i,j}[\mu_{ij}\log\mu_{ij}+(1-\mu_{ij})\log(1-\mu_{ij})]\Bigg\},
\end{align*}

We will study the limit as $n\rightarrow\infty$. 
Before we proceed, we need a representation
of the vector $\alpha$ in the infinite network. The following
assumption guarantee that we can switch from the discrete
to the continuum.
\begin{axiom}\label{AssumpI}
Assume that
\begin{equation*}
\alpha_{ij}=\alpha\left(i/n,j/n\right),
\end{equation*}
where $\alpha(x,y):[0,1]^{2}\rightarrow\mathbb{R}$ is a deterministic exogenous function that is symmetric,
i.e., $\alpha(x,y)=\alpha(y,x)$.
\footnote{To ease the notations, we project $\otimes_{j=1}^{S}\mathcal{X}_{j}$ onto $[0,1]$
and the function $\alpha(\tau_{i},\tau_{j})$ defined previously is now re-defined from $[0,1]^{2}$ to $\mathbb{R}$.}
\end{axiom}

Since we have $n$ players, the number of types for the players must be finite, although
it may grow as $n$ grows. $\alpha_{ij}$ are symmetric, and can take at most $\frac{n(n+1)}{2}$ values.
As $n\rightarrow\infty$, the number of types can become infinite and
$\alpha(x,y)$ may take infinitely many values. On the other hand, in terms of practical applications,
finitely many values often suffice
\footnote{If an entry of the vector $\tau_{i}$ is
continuous, we can always transform the variable in a discrete vector
using thresholds. For example, if $\mathcal{X}_{j}=\text{[\$50,000,\$200,000]}$, 
we can bucket the incomes into three levels, low: [\$50,000,\$100,000), medium [\$100,000,\$150,000) and high: [\$150,000,
\$200,000].}.

\begin{axiom}\label{AssumpII}
We assume that $\alpha(x,y)$ is uniformly bounded in $x$ and $y$:
\begin{equation}
\sup_{(x,y)\in[0,1]^{2}}|\alpha(x,y)|<\infty.
\end{equation}
\end{axiom}

\begin{figure}
\caption{Examples of function $\alpha(x,y)$.}
\centering
\includegraphics[scale=.35]{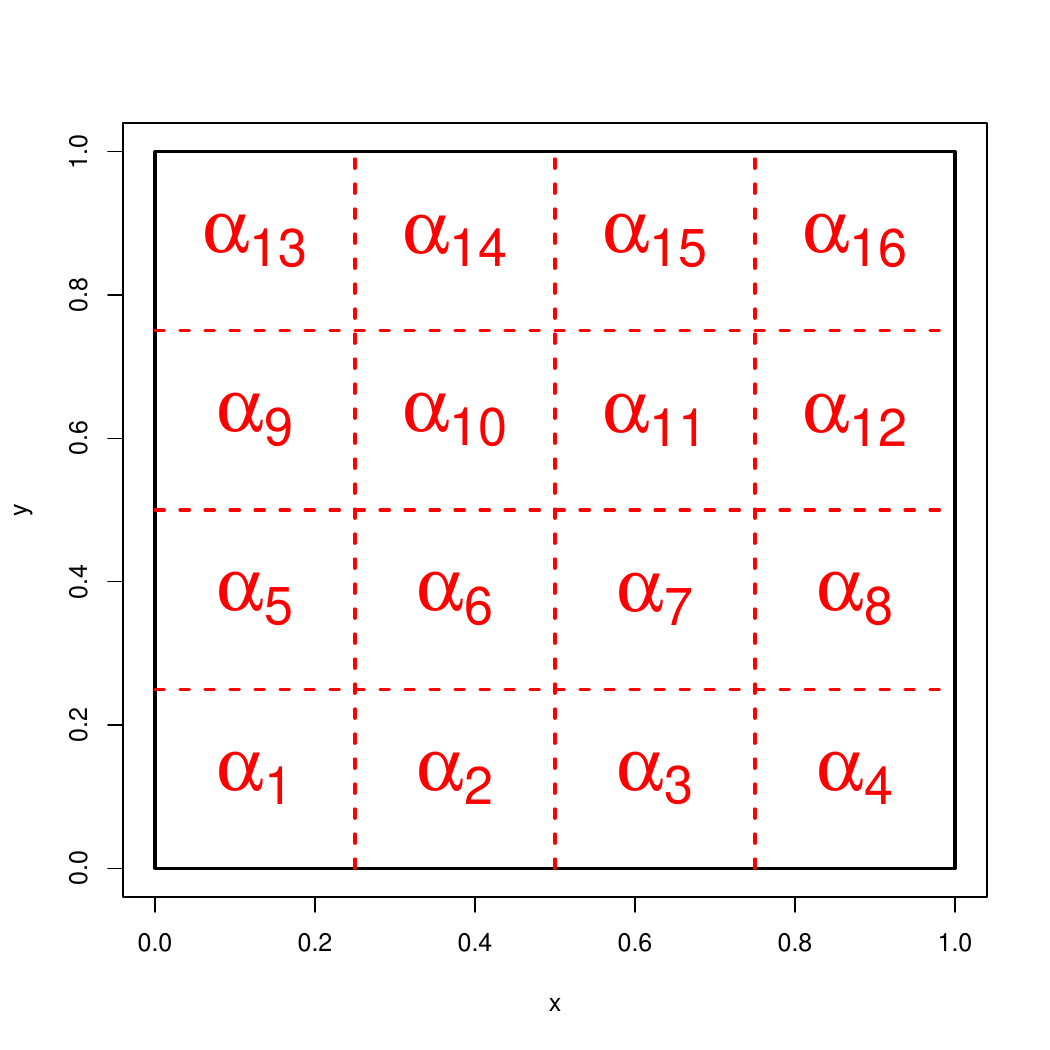}
\includegraphics[scale=.35]{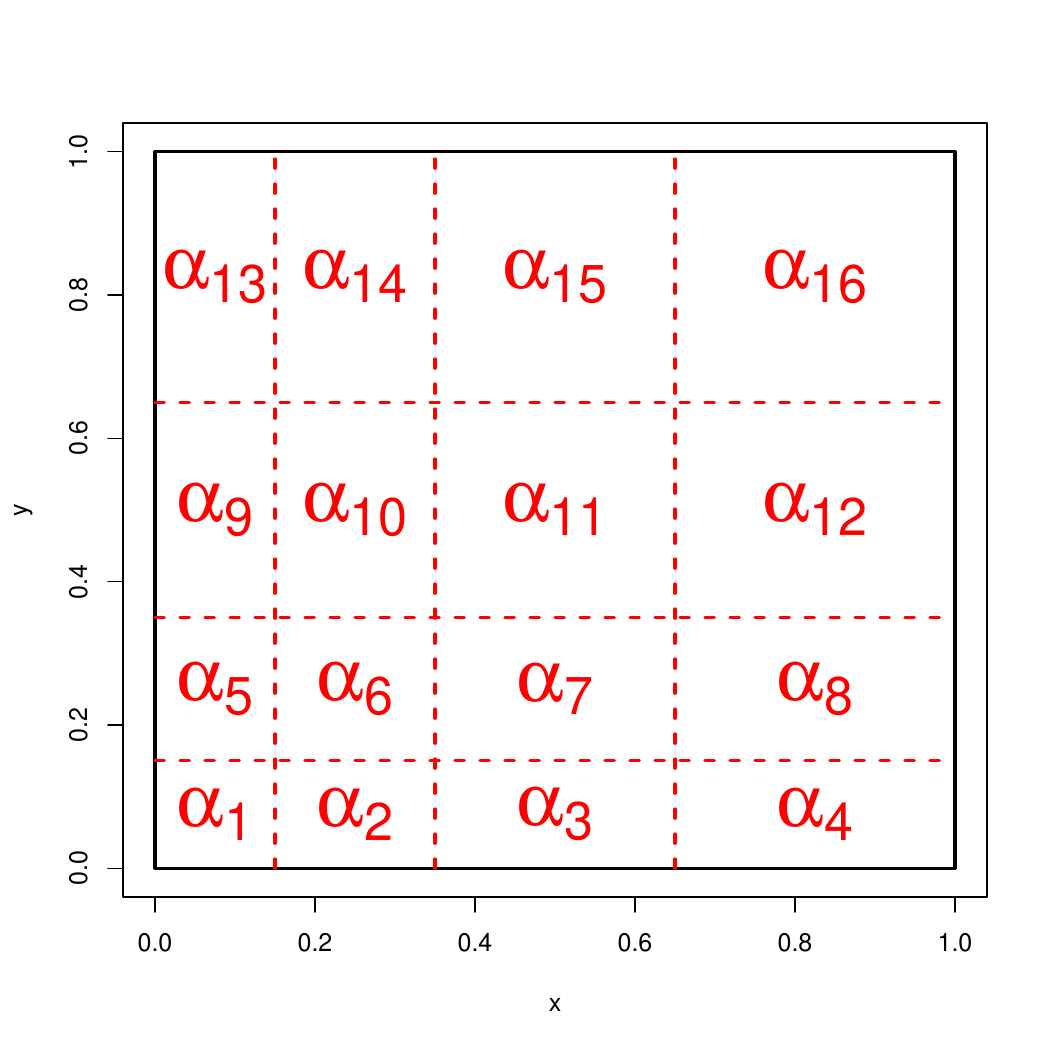}\\
(A)\hspace{6cm}(B)\\
\includegraphics[scale=.35]{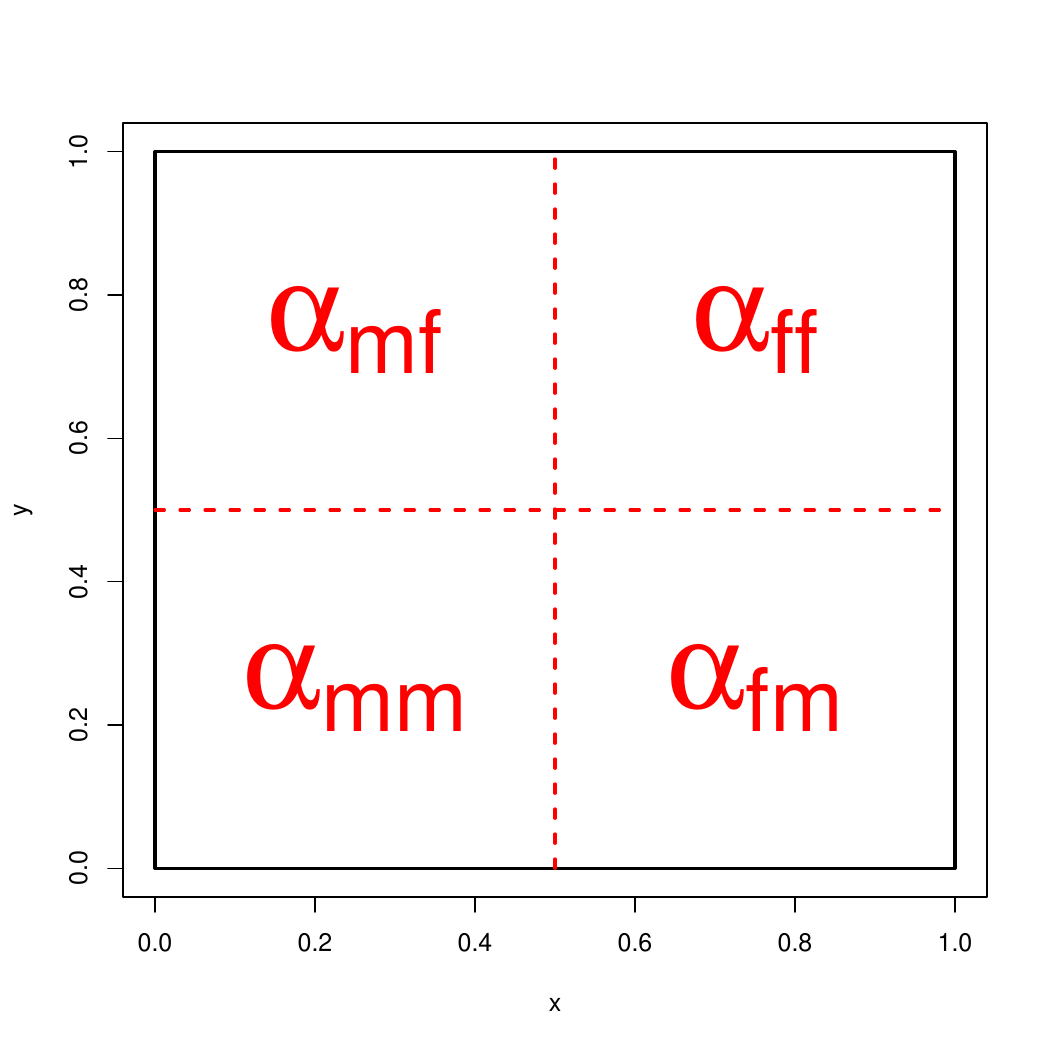}
\includegraphics[scale=.35]{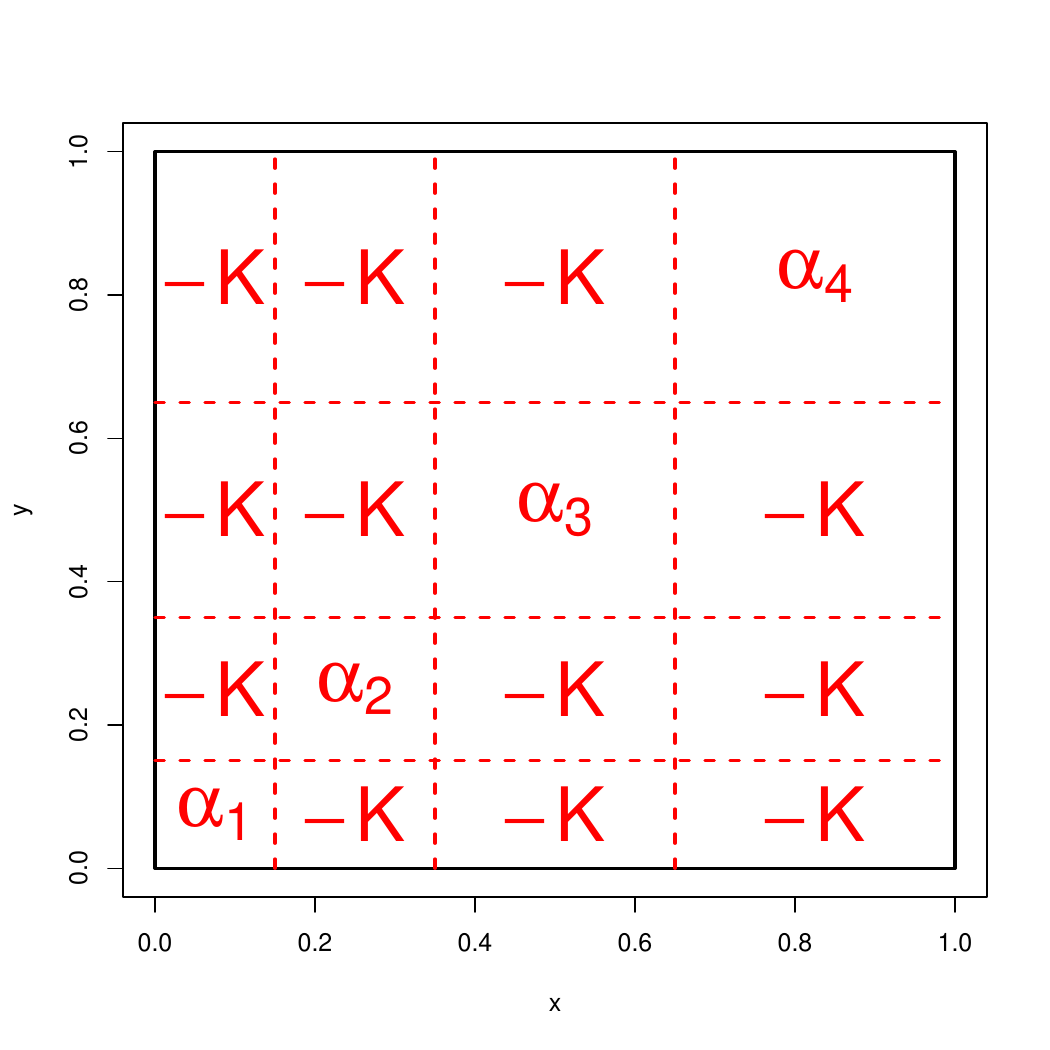}\\
(C)\hspace{6cm}(D)\\
\flushleft The figure provides several examples of possible partitions of the net benefit function $\alpha(x,y)$ with finite covariates.
The asymptotic version of this function is defined over the unit square. 
\label{fig:partitionsalpha}
\end{figure}
As a simple example, let us consider gender: the population consists
of males and female agents. 
For example, half of the nodes (population)
are males, say $i=1,2,\ldots,\frac{n}{2}$
and the other half are females, $i=\frac{n}{2}+1,\frac{n}{2}+2,\ldots,n$.\footnote{Here, we assume without loss
of generality that $n$ is an even number.} That means, $\alpha(x,y)$ takes three values according
to the three regions:
\begin{eqnarray*}
&\left\{(x,y): 0<x,y<\frac{1}{2}\right\},
\\
&\left\{(x,y): \frac{1}{2}<x,y<1\right\},
\\
&\left\{(x,y):0<x<\frac{1}{2}<y<1\right\}\bigcup
\left\{(x,y):0<y<\frac{1}{2}<x<1\right\},
\end{eqnarray*}
and these three regions correspond precisely to pairs: male-male,
female-female, and male-female. This example is represented in Figure \ref{fig:partitionsalpha}(C).\\

The work of \citet{DiaconisChatterjee2011} show that the variational problem in (\ref{eqn:varapproxgeneral}) translates into an analogous variational problem
for the graph limit.\footnote{See also \citet{Mele2010a}  for similar results in a directed network.} 
In the special case $\alpha(x,y)\equiv\alpha$, 
it is shown in \citet{DiaconisChatterjee2011} that as
$n\rightarrow\infty$ the log-constant of the ERGM converges
to the solution of the variational problem (\ref{eqn:varproblem_CD2013}),
that is
\begin{equation}\label{psinpsi}
\psi_{n}(\alpha,\beta,\gamma)\rightarrow\psi(\alpha,\beta,\gamma),
\end{equation}
where
\begin{align}
\psi(\alpha,\beta,\gamma) &=\sup_{h\in\mathcal{W}}
\bigg\{\alpha\int_{0}^{1}\int_{0}^{1}h(x,y)dxdy+\frac{\beta}{2}\int_{0}^{1}\int_{0}^{1}\int_{0}^{1}h(x,y)h(y,z)dxdydz
\label{eqn:varproblem_CD2013}
\\
&\qquad
+\frac{2\gamma}{3}\int_{0}^{1}\int_{0}^{1}\int_{0}^{1}h(x,y)h(y,z)h(z,x)dxdydz
- \frac{1}{2}\int_{0}^{1}\int_{0}^{1}I(h(x,y))dxdy\bigg\},
\nonumber
\end{align}
where 
\begin{equation}
\mathcal{W}:=\left\{h:[0,1]^{2}\rightarrow[0,1],h(x,y)=h(y,x), 0\leq x,y\leq 1\right\},
\end{equation}
and we define the entropy function:
\begin{equation*}
I(x):=x\log x+(1-x)\log(1-x),\qquad 0\leq x\leq 1,
\end{equation*}
with $I(0)=I(1)=0$.

In essence the first three terms in \eqref{eqn:varproblem_CD2013} correspond to the expected potential function in the continuum, while the
last term in \eqref{eqn:varproblem_CD2013} corresponds to the entropy of the graph limit.

We will show that \eqref{psinpsi} holds with 
\begin{align}\label{psi:general}
\psi(\alpha,\beta,\gamma) &=\sup_{h\in\mathcal{W}}
\bigg\{\int_{0}^{1}\int_{0}^{1}\alpha(x,y)h(x,y)dxdy+\frac{\beta}{2}\int_{0}^{1}\int_{0}^{1}\int_{0}^{1}h(x,y)h(y,z)dxdydz
\\
&\qquad
+\frac{2\gamma}{3}\int_{0}^{1}\int_{0}^{1}\int_{0}^{1}h(x,y)h(y,z)h(z,x)dxdydz
- \frac{1}{2}\int_{0}^{1}\int_{0}^{1}I(h(x,y))dxdy\bigg\},
\nonumber
\end{align}

The function $h$ in the expressions above is known
as the \emph{graphon} from the graph limits literature
\footnote{See \citet{LovaszBook2012}, \citet{BorgsEtAl2008}}, 
large deviations literature for random graphs\footnote{See \citet{ChatterjeeVaradhan2011},
\citet{DiaconisChatterjee2011}} and analysis of the resulting variational problem.\footnote{See
\citet{AristoffZhu2014}, \citet{RadinYin2013} among others.}
and it is a
representation of an infinite network, 
where $h$ is a simple symmetric function $h:[0,1]^{2}\rightarrow[0,1]$,
and $h(x,y)=h(y,x)$.
Note that our goal is to approximate $\psi_{n}^{MF}$
and hence $\psi_{n}$ by $\psi$, whose definition 
involves the function $h$, and we call such a function a graphon
in the rest of the paper,
to be consistent with the literature, while we are not attempting
here to establish a theory of graph limits to allow nodal
covariates. That is an interesting research direction
worth investigating in the future, but is out of the scope
of the current paper.

The following proposition shows that for a model with finitely many types
the variational approximation is asymptotically exact.

\begin{proposition}
Under Assumptions \ref{AssumpI} and \ref{AssumpII}, as $n\rightarrow\infty$
\begin{equation*}
\psi_{n}(\alpha,\beta,\gamma) \rightarrow \psi(\alpha,\beta,\gamma),
\end{equation*}
where $\psi(\alpha,\beta,\gamma)$ is defined in \eqref{psi:general}.
\end{proposition}

\begin{proof}
It follows directly from Theorem \ref{thm:mf_bounds} and 
$\psi_{n}^{MF}(\bm{\mu}(\alpha,\beta,\gamma)) \rightarrow \psi(\alpha,\beta,\gamma)$, 
as $n\rightarrow\infty$.
\end{proof}

The proposition states that as $n$ becomes large, we can approximate the
exponential random graph using a model with independent
links (conditional on finitely many types). This is a very useful result
because the latter approximation is simple and tractable, while the exponential random
graph model contains complex dependence patterns that make estimation
computationally expensive. \\

\vspace{1cm}

%%%%%%%%%%%%%%%%%%%%%%%%%%%%%%%%%%%%%%%%%%%%%%%%%
\subsection{Approximation of the limit log normalizing constant}
We can analyze and provide an
approximation of the log-constant in the large network limit. 
The variational formula for $\psi(\alpha,\beta,\gamma)$ is an infinite-dimensional
problem which is intractable in most cases. 
Nevertheless, we can always bound the infinite dimensional problem with finite
dimensional ones (both lower and upper bounds), at least
in the absence of transitivity.
For details, see Proposition \ref{prop:boundsConstant} in the Online Appendix.
The lower-bound in Proposition \ref{prop:boundsConstant} coincides with the structured mean-field approach
of \cite{XingJordanRussell2002}. In a model with
homogeneous players, the lower-bound corresponds to the computational
approximation of graph limits implemented in \cite{ZhengHe2013}.

In the case of extreme homophily, we can also obtain 
finite-dimensional approximation, see Proposition \ref{negativeProp}
in the Online Appendix.

%%%%%%%%%%%%%%%%%%%%%%%%%%%%%%%%%%%%%%%%%%%%%%%%%%%%%
\subsection{Characterization of the variational problem}\label{section:char_vp}

We recall that the log normalizing constant in 
the $n\rightarrow\infty$ limit is given by the variational problem:
\begin{align}\label{VP}
\psi(\alpha,\beta,\gamma)
&=\sup_{h\in\mathcal{W}}
\bigg\{\int_{0}^{1}\int_{0}^{1}\alpha(x,y)h(x,y)dxdy+\frac{\beta}{2}\int_{0}^{1}\int_{0}^{1}\int_{0}^{1}h(x,y)h(y,z)dxdydz
\\
&\qquad\qquad\qquad
+\frac{2\gamma}{3}\int_{0}^{1}\int_{0}^{1}\int_{0}^{1}h(x,y)h(y,z)h(z,x)dxdydz
\nonumber
\\
&\qquad
-\frac{1}{2}\int_{0}^{1}\int_{0}^{1}\left[h(x,y)\log h(x,y)+(1-h(x,y))\log(1-h(x,y))\right]dxdy\bigg\}.
\nonumber
\end{align}

\begin{proposition}
The optimal graphon $h$ that solves the variational problem \eqref{VP} satisfies the Euler-Lagrange equation:
\begin{equation}\label{EL}
2\alpha(x,y)+\beta\int_{0}^{1}h(x,y)dx+\beta\int_{0}^{1}h(x,y)dy
+4\gamma\int_{0}^{1}h(x,z)h(y,z)dz
=\log\left(\frac{h(x,y)}{1-h(x,y)}\right).
\end{equation}
\end{proposition}

\begin{proof}
The proof follows from the same argument as in Theorem 6.1. in \cite{DiaconisChatterjee2011}.
\end{proof}

\begin{corollary}
If $\alpha(x,y)$ is not a constant function, then the optimal graphon $h$ that solves the variational problem \eqref{VP}
is not a constant function.
\end{corollary}

\begin{proof}
If the optimal graphon $h$ is a constant function, then \eqref{EL} implies that $\alpha$ is a constant function.
Contradiction.
\end{proof}

In general, if a graphon satisfies the Euler-Lagrange equation, that only indicates
that the graphon is a stationary point, and it is not clear if the graphon is the local maximizer, local minimizer or neither.
In the next result, we will show that when $\beta$ is negative, any graphon satisfying the Euler-Lagrange equation in our model
is indeed a local maximizer.

\begin{proposition}
Assume that $\beta<0$ and $\gamma=0$. If $h$ is a graphon that satisfies the Euler-Lagrange equation \eqref{EL}, then $h$ is a local maximizer
of the variational problem \eqref{VP}.
\end{proposition}

\begin{proof}
Let us define
\begin{align}
\Lambda[h]&:=\int_{0}^{1}\int_{0}^{1}\alpha(x,y)h(x,y)dxdy+\frac{\beta}{2}\int_{0}^{1}\int_{0}^{1}\int_{0}^{1}h(x,y)h(y,z)dxdydz
\\
&\qquad
-\frac{1}{2}\int_{0}^{1}\int_{0}^{1}\left[h(x,y)\log h(x,y)+(1-h(x,y))\log(1-h(x,y))\right]dxdy.
\nonumber
\end{align}
Let $h$ satisfy \eqref{EL} and
for any symmetric function $g$ and $\epsilon>0$ sufficiently small,
we have
\begin{align}
&\Lambda[h+\epsilon g]-\Lambda[h]
\\
&=\epsilon^{2}\left[\frac{\beta}{2}\int_{0}^{1}\left(\int_{0}^{1}g(x,y)dy\right)^{2}dx
-\frac{1}{4}\int_{0}^{1}\int_{0}^{1}I''(h(x,y))g^{2}(x,y)dxdy\right]+O(\epsilon^{3})
\nonumber
\\
&=\epsilon^{2}\left[\frac{\beta}{2}\int_{0}^{1}\left(\int_{0}^{1}g(x,y)dy\right)^{2}dx
-\frac{1}{4}\int_{0}^{1}\int_{0}^{1}\frac{g^{2}(x,y)}{h(x,y)(1-h(x,y))}dxdy\right]+O(\epsilon^{3}),
\nonumber
\end{align}
and since $\beta<0$, we conclude that $h$ is a local maximizer in \eqref{VP}.
\end{proof}

\begin{remark}\label{beta0}
In general, the variational problem for the graphons
and the corresponding Euler-Lagrange equation \eqref{EL}
does not yield a closed form solution. In the special case $\beta=\gamma=0$,
\begin{equation}
\psi(\alpha,0,0)=\sup_{h\in\mathcal{W}}\left\{ \iint_{[0,1]^{2}}\alpha(x,y)h(x,y)dxdy-\frac{1}{2}\iint_{[0,1]^{2}}I(h(x,y))dxdy\right\} ,
\end{equation}
where $I(x):=x\log x+(1-x)\log(1-x)$ and it is easy to see that the
optimal graphon $h(x,y)$ is given by
$h(x,y)=\frac{e^{2\alpha(x,y)}}{e^{2\alpha(x,y)}+1}$,
and therefore, $\psi(\alpha,0,0)=\frac{1}{2}\iint_{[0,1]^{2}}\log(1+e^{2\alpha(x,y)})dxdy$.
\end{remark}

\section{Details of Equilibrium Economic Foundations}\label{section:appendix_microfoundation}

\subsection{Setup and preferences}
Consider a population of $n$ heterogeneous players (the nodes), each
characterized by an exogenous
type $\tau_i\in \otimes_{j=1}^{S}\mathcal{X}_{j}$, $i=1,...,n$. 
The attribute $\tau_i$ is an $S$-dimensional
vector and the sets $\mathcal{X}_{j}$ can represent age, race, gender,
income, etc. 
\footnote{For instance, if we consider gender and income, 
then $S=2$, and we can take $\otimes_{j=1}^{2}\mathcal{X}_{j}=\{\text{male,female}\}\times\{\text{low, medium, high}\}$.
The sets $\mathcal{X}_{j}$ can be both discrete and continuous. For example, if we consider gender and income, 
we can also take $\otimes_{j=1}^{2}\mathcal{X}_{j}=\{\text{male,female}\}\times\text{[\$50,000,\$200,000]}$. Below we restrict the covariates to be discrete, but we allow the number of types to grow with the size of the network.}
We collect all $\tau_i$'s in an $n\times S$ matrix $\tau$. The network's
 adjacency matrix $g$ has entries
$g_{ij}=1$ if $i$ and $j$ are linked; and $g_{ij}=0$ otherwise.
The network is undirected, i.e. $g_{ij}=g_{ji}$, and $g_{ii}=0$,
for all $i$'s.\footnote{Extensions to directed networks
are straightforward (see \citet{Mele2010a}).} The utility of player $i$ is
\begin{equation}
\label{eq:utility_fcn_appendix}
u_{i}(g,\tau)=\sum_{j=1}^{n}\alpha_{ij}g_{ij}
+\frac{\beta}{n}\sum_{j=1}^{n}\sum_{k=1}^{n}g_{ij}g_{jk},
\end{equation}
where $\alpha_{ij}:=\nu(\tau_i, \tau_j)$ are symmetric functions
$\nu:\otimes_{j=1}^{S}\mathcal{X}_{j}\times\otimes_{j=1}^{S}\mathcal{X}_{j}\rightarrow\mathbb{R}$
and $\nu(\tau_i,\tau_j)=\nu(\tau_j,\tau_i)$ for all $i,j$;
and $\beta$ is a scalar.
The utility of player $i$
depends on the number of direct links, each weighted according to a function $\nu$ of the types $\tau$. 
This payoff structure implies that the net benefit of forming a direct connection depends on the characteristics of the two individuals 
involved in the link. 

Players also care about the number of links that each of their
direct contacts have formed.\footnote{The normalization of $\beta$ by $n$
is necessary for the asymptotic analysis.} For example, when $\beta>0$, there is an incentive to form links to people that have many friends, e.g. popular kids in school. On the other hand, when $\beta<0$ the incentive is reversed. For example, one can think that forming links to a person with many connections could decrease our visibility and decrease the effectiveness of interactions. Similar utility functions
have been used extensively in the empirical network formation
literature.\footnote{See \cite{Mele2010a}, \cite{Sheng2012},
\cite{DePaulaEtAl2012}, \cite{ChandrasekharJackson2012}, \cite{Badev2013},
\cite{Butts2009}. }

The preferences in \eqref{eq:utility_fcn_appendix} include only direct links and friends' populatity. However, we can also include other types of link externalities. For example, in many applications the researcher is interested in estimating preferences for common neighbors. This is an important network statistics to measure transitity and clustering in networks. In our model we can easily add an utility component to capture these effects.
\begin{equation}
\label{eq:utility_fcn_triangles_appendix}
u_{i}(g,\tau)=\sum_{j=1}^{n}\alpha_{ij}g_{ij}
+\frac{\beta}{n}\sum_{j=1}^{n}\sum_{k=1}^{n}g_{ij}g_{jk} 
+ \frac{\gamma}{n}\sum_{j=1}^{n}\sum_{k=1}^{n}g_{ij}g_{jk}g_{ki},
\end{equation}
These preferences include an additional parameter $\gamma$ that measures the effect of common neighbors. The potential function for this model is 
\begin{equation}
\label{eq:potential_triangles_appendix}
Q_{n}(g;\alpha,\beta)=\sum_{i=1}^{n}\sum_{j=1}^{n}\alpha_{ij}g_{ij}+\frac{\beta}{2n}\sum_{i=1}^{n}\sum_{j=1}^{n}\sum_{k=1}^{n}g_{ij}g_{jk} 
+ \frac{2\gamma}{3n}\sum_{j=1}^{n}\sum_{k=1}^{n}g_{ij}g_{jk}g_{ki} .
\end{equation}
In general, all the results that we show below extend to more general utility functions that include payoffs for link externalities similar to (\ref{eq:utility_fcn_triangles}).

The probability that $i$ and $j$ meet can depend on their networks: it could be a function of their common neighbors, or a function of their degrees and centralities, for example. In Assumption \ref{assm:meeting}, we assume that the existence of a link between $i$ and $j$ does not affect their probability of meeting. This is because we prove the existence and functional form of the stationary distribution (\ref{eq:stat_distribution}) using the detailed balance condition, which is not satisfied if we allow the meeting probabilities to depend  on the link between $i$ and $j$.  

The model can easily be extended to directed networks and the results on equilibria and long-run stationary distribution will hold. The results about the approximations of the likelihood shown below will also hold for directed networks, with minimal modifications of the proofs.

Finally, while our model generates dense graphs, the approximations using variational methods and nonlinear large deviations that we develop in the rest of the paper also work in moderately sparse graphs.
More precisely, the utility of player $i$ is given by
\begin{equation}
u_{i}(g,\tau)=\sum_{j=1}^{n}\alpha_{ij}^{(n)}g_{ij}
+\frac{\beta^{(n)}}{n}\sum_{j=1}^{n}\sum_{k=1}^{n}g_{ij}g_{jk}
+\frac{\gamma^{(n)}}{n}\sum_{j=1}^{n}\sum_{k=1}^{n}g_{ij}g_{jk}g_{ki},
\end{equation}
where $|\alpha_{ij}^{(n)}|$, $|\beta^{(n)}|$ and $|\gamma^{(n)}|$ can have moderate growth in $n$ instead of being bounded. 
We will give more details later in our paper. 
\footnote{See \cite{ChatterjeeDembo2014} for additional applications of nonlinear large deviations.}

\begin{example}\label{example:homophily}(Homophily)
Consider a model with $\nu(\tau_i,\tau_j)=V-c(\tau_i,\tau_j)$, where $V>0$ is
the benefit of a link and $c(\tau_i,\tau_j)$ ($=c(\tau_j,\tau_i)$) is
the cost of the link
between $i$ and $j$. To model homophily in this
framework let the cost function be
\begin{equation}
c(\tau_i,\tau_j) = \begin{cases}
c & \mbox{if \ensuremath{\tau_{i}=\tau_{j}}},\\
C & \mbox{if \ensuremath{\tau_{i}\neq\tau_{j}}}.
\end{cases}
\end{equation}
For example, consider the parameterization
$0<c<V<C$ and $\beta=0$, $\gamma=0$. In this case the players have no incentive to form links with agents of other groups. On the other hand, if we have
$0<c<V<C$ and $\beta,\gamma>0$, also links across groups will be formed, as
long as $\beta,\gamma$ are sufficiently large.
\end{example}

\begin{example}\label{example:social_distance} (Social Distance Model)
Let the payoff from direct links be a function of the social distance
among the individuals. Formally, let $\nu(\tau_i,\tau_j) := \eta d(\tau_i,\tau_j) - c$,
where $d(\tau_i,\tau_j)$ is a distance function, $\eta$ is a
parameter that determines the sensitivity to the social distance
and $c>0$ is the cost of forming a link.\footnote{See \cite{IijimaKamada2014} for
a more general example of such model.} The case with $\eta<0$ represents
a world where individuals prefer linking to similar agents
and $\eta>0$ represents a world where individuals prefer linking
with people at larger social distance.
Note that even when $\eta<0$, if we have $\beta,\gamma>0$ sufficiently large, individuals may still have an incentive
to form links with people at larger social distance.
\end{example}

\subsection{Meetings and equilibrium}
The network formation process follows a stochastic best-response
dynamics:\footnote{See \cite{Blume1993}, \cite{Mele2010a}, \cite{Badev2013}.}
in each period $t$, two random players meet with probability $\rho_{ij}$;
upon meeting they have the opportunity to form a link (or delete it,
if already in place). Players are myopic:
when they form a new link, they do not consider
how the new link will affect the incentives of the other
player in the future evolution of the network.\\

\begin{axiom}\label{assm:meeting}
The meeting process is a function of types and the network. Let $g_{-ij}$ indicate the network $g$ without considering the link $g_{ij}$. Then the probability that $i$ and $j$ meet is
\begin{equation}
\rho_{ij} := \rho(\tau_i,\tau_j, g_{-ij}) > 0
\end{equation}  
for all pairs $i$ and $j$,  and  i.i.d. over time.
\end{axiom}
Assumption \ref{assm:meeting} implies that the meeting process
can depend on covariates and the state of the network. For example, if two players have many friends in common they may meet with high probability; or people that share some demographics may meet more often. Crucially, every pair of players has a strictly
positive probability of meeting. This guarantees that
each link of the network has the opportunity of being revised.\\ 
\indent Upon meetings, players decide whether to form or delete a link by
maximizing the sum of their current utilities, i.e. the total
surplus generated by the relationship. We are implicitly
assuming that individuals can transfer utilities. When deciding
whether to form a new link or deleting an existing
link, players receive a random matching shock $\varepsilon_{ij}$
that shifts their preferences.\\
\indent At time $t$, the links $g_{ij}$ is formed if
\begin{equation*}
u_i(g_{ij}=1, g_{-ij},\tau) + u_j(g_{ij}=1,g_{-ij},\tau) + \varepsilon_{ij}(1)\geq
u_i(g_{ij}=0, g_{-ij},\tau) + u_j(g_{ij}=0,g_{-ij},\tau) + \varepsilon_{ij}(0)\,.
\end{equation*}
We make the following assumptions on the matching value.
\begin{axiom}\label{assm:logisticshocks}Individuals receive a
\textit{logistic shock} before they decide whether to form a link (i.i.d. over time
and players).\end{axiom}
The logistic assumption is standard in many discrete choice models in economics and statistics (\cite{Train2009}). \\
\indent We can now characterize the equilibria of the model, following
 \cite{Mele2010a} and \cite{ChandrasekharJackson2012}.
In particular, we can show that the
network formation is a
potential game (\citet{MondererShapley1996}).
\begin{proposition}\label{prop:potential_appendix}
The network formation is a potential game, and
there exists a potential function $Q_{n}(g;\alpha,\beta)$
that characterizes the incentives of all the players in any state
of the network
\begin{equation}
\label{eq:potential_appendix}
Q_{n}(g;\alpha,\beta)=\sum_{i=1}^{n}\sum_{j=1}^{n}\alpha_{ij}g_{ij}+\frac{\beta}{2n}\sum_{i=1}^{n}\sum_{j=1}^{n}\sum_{k=1}^{n}g_{ij}g_{jk}
+\frac{2\gamma}{3n}\sum_{i=1}^{n}\sum_{j=1}^{n}\sum_{k=1}^{n}g_{ij}g_{jk}g_{ki}.
\end{equation}
\end{proposition}

\begin{proof}
The proposition follows the same lines as Proposition 1
in \cite{Mele2010a} and it is omitted for brevity.
\end{proof}

The potential function $Q_{n}(g;\alpha,\beta)$ is such that, for any $g_{ij}$
\[
Q_{n}(g;\alpha,\beta)-Q_{n}(g-ij;\alpha,\beta)=u_{i}(g)+u_{j}(g)-\left[u_{i}(g-ij)+u_{j}(g-ij)\right].
\]
Thus we can keep track of all players' incentives using
the scalar $Q_{n}(g;\alpha,\beta)$. It is easy to show
 that all the pairwise stable
(with transfers) networks are the local maxima of the potential
function.\footnote{A network $g$ is pairwise stable with transfers if:
(1) $g_{ij}=1\Rightarrow u_{i}(g,\tau)+u_{j}(g,\tau)\geq u_{i}(g-ij,\tau)+u_{j}(g-ij,\tau)$
and (2) $g_{ij}=0\Rightarrow u_{i}(g,\tau)+u_{j}(g,\tau)\geq u_{i}(g+ij,\tau)+u_{j}(g+ij,\tau)$;
where $g+ij$ represents network $g$ with the addition of link $g_{ij}$
and network $g-ij$ represents network $g$ without link $g_{ij}$. See
\citet{Jackson2008} for more details.} The sequential network formation
follows a \textit{Glauber} dynamics, therefore converging
to a unique stationary distribution.
\begin{theorem}
\label{thm:uniquestatdist}In the long run, the model converges to
the stationary distribution $\pi_{n}$, defined as
\begin{equation}
\label{eq:stat_distribution_appendix}
\pi_{n}(g;\alpha,\beta)=\frac{\exp\left[Q_{n}(g;\alpha,\beta)\right]}
{\sum_{\omega\in\mathcal{G}}\exp\left[Q_{n}(\omega;\alpha,\beta)\right]}
=\exp\left\{ n^{2}\left[T_{n}(g;\alpha,\beta)-\psi_{n}(\alpha,\beta)\right]\right\},
\end{equation}
where $T_{n}(g;\alpha,\beta)=n^{-2} Q_{n}(g;\alpha,\beta)$,
\begin{equation}\label{psi_n_appendix}
\psi_{n}(\alpha,\beta)=\frac{1}{n^{2}}\log\sum_{\omega\in\mathcal{G}}\exp\left[n^{2}T_{n}(\omega;\alpha,\beta)\right],
\end{equation}
and $\mathcal{G}:=\{\omega=(\omega_{ij})_{1\leq i,j\leq n}:\omega_{ij}=\omega_{ji}\in\{0,1\}, \omega_{ii}=0, 1\leq i,j\leq n\}$.
\end{theorem}
\begin{proof} The proof is an extension of Theorem 1 in \cite{Mele2010a}.
See also \cite{ChandrasekharJackson2012} and \cite{Butts2009}.
\end{proof}
Notice that the likelihood (\ref{eq:stat_distribution})
corresponds to an ERGM model with heterogeneous nodes and two-stars.
As a consequence our model inherits all the estimation and identification
challenges of the ERGM model.\\

%%%%%%%%%%%%%%%%%%%%%%%%%%%%%%%

\section{Special Case: The Edge-Star Model}

The general solution of the variational problem (\ref{eqn:varproblem_CD2013}) is complicated. However, there are some special cases where we can characterize the solution with extreme detail. These examples show how we can solve the variational approximation in stylized settings, and we use them to explain how the method works in practice.
In this section, we consider the special case in the absence
of transitivity, i.e. $\gamma=0$ and we get further results
for the edge-star model.

\subsection{Extreme homophily}

We can exploit homophily to obtain a tractable
approximation. Suppose that there are $M$ types in
the population. The cost of forming
links among individuals of the same group is finite, but there is
a large cost of forming links among people of different groups (potentially infinite).
We show that in this case the normalizing constant can be approximated
by solving $M$ independent univariate maximization problems.
In the special case of extreme homophily, our model converges to a block-diagonal model.

\begin{proposition}\label{negativeProp} Let $0=a_{0}<a_{1}<\cdots<a_{M}=1$
be a given sequence. Assume that
\begin{equation}
\alpha(x,y)=\alpha_{mm},\qquad\text{if }a_{m-1}<x,y<a_{m},\qquad m=1,2,\ldots,M.
\end{equation}
and $\alpha(x,y)\leq-K$ otherwise is a given function. Let $\psi(\alpha,\beta,0;-K)$ be
the variational problem for the graphons and
$\psi(\alpha,\beta,0;-\infty)=\lim_{K\rightarrow\infty}\psi(\alpha,\beta,0;-K)$.
Then, we have
\begin{equation}
\psi(\alpha,\beta,0;-\infty)=\sum_{m=1}^{M}(a_{m}-a_{m-1})^{2}\sup_{0\leq x\leq1}\left\{ \alpha_{mm}x+\frac{\beta}{2}x^{2}-\frac{1}{2}I(x)\right\} .
\label{eq:extrhomophily}
\end{equation}
\end{proposition}

\begin{proof}
First, observe that
\begin{align}
 & \psi(\alpha,\beta,0;-\infty)\\
 & =\sup_{h\in\mathcal{W}^{-}}\bigg\{\sum_{i=1}^{M}\alpha_{i}\iint_{[a_{i-1},a_{i}]^{2}}h(x,y)dxdy+\frac{\beta}{2}\int_{0}^{1}\int_{0}^{1}h(x,y)h(y,z)dxdydz\nonumber \\
 & \qquad\qquad\qquad-\frac{1}{2}\sum_{i=1}^{M}\iint_{[a_{i-1},a_{i}]^{2}}I(h(x,y))dxdy\bigg\}\nonumber \\
 & =\sup_{h\in\mathcal{W}^{-}}\bigg\{\sum_{i=1}^{M}\alpha_{i}\iint_{[a_{i-1},a_{i}]^{2}}h(x,y)dxdy+\frac{\beta}{2}\sum_{i=1}^{M}\int_{a_{i-1}}^{a_{i}}\left(\int_{a_{i-1}}^{a_{i}}h(x,y)dy\right)^{2}dx\nonumber \\
 & \qquad\qquad\qquad-\frac{1}{2}\sum_{i=1}^{M}\iint_{[a_{i-1},a_{i}]^{2}}I(h(x,y))dxdy\bigg\}\nonumber \\
 & =\sum_{i=1}^{M}\sup_{\substack{h:[a_{i-1},a_{i}]^{2}\rightarrow[0,1]\\
h(x,y)=h(y,x)
}
}\bigg\{\alpha_{i}\iint_{[a_{i-1},a_{i}]^{2}}h(x,y)dxdy+\frac{\beta}{2}\int_{a_{i-1}}^{a_{i}}\left(\int_{a_{i-1}}^{a_{i}}h(x,y)dy\right)^{2}dx\nonumber \\
 & \qquad\qquad\qquad-\frac{1}{2}\iint_{[a_{i-1},a_{i}]^{2}}I(h(x,y))dxdy\bigg\},\nonumber
\end{align}
where
\begin{equation}
\mathcal{W}^{-}:=\left\{ h\in\mathcal{W}:h(x,y)=0\text{ for any }(x,y)\notin\bigcup_{i=1}^{M}[a_{i-1},a_{i}]^{2}\right\} .
\end{equation}
By taking $h$ to be a constant on $[a_{i-1},a_{i}]^{2}$, it is clear
that
\begin{equation}
\psi(\alpha,\beta,0;-\infty)\geq\sum_{i=1}^{M}(a_{i}-a_{i-1})^{2}\sup_{0\leq x\leq1}\left\{ \alpha_{i}x+\frac{\beta}{2}x^{2}-\frac{1}{2}I(x)\right\} .
\end{equation}
By Jensen's inequality
\begin{align}
\psi(\alpha,\beta,0;-\infty) & \leq\sum_{i=1}^{M}\sup_{\substack{h:[a_{i-1},a_{i}]^{2}\rightarrow[0,1]\\
h(x,y)=h(y,x)
}
}\bigg\{\alpha_{i}\int_{a_{i-1}}^{a_{i}}\left(\int_{a_{i-1}}^{a_{i}}h(x,y)dy\right)dx\\
 & \qquad\qquad+\frac{\beta}{2}\int_{a_{i-1}}^{a_{i}}\left(\int_{a_{i-1}}^{a_{i}}h(x,y)dy\right)^{2}dx\nonumber \\
 & \qquad\qquad\qquad-\frac{1}{2}(a_{i}-a_{i-1})\int_{a_{i-1}}^{a_{i}}I\left(\frac{1}{a_{i}-a_{i-1}}\int_{a_{i-1}}^{a_{i}}h(x,y)dy\right)dx\bigg\}\nonumber \\
 & \leq\sum_{i=1}^{M}(a_{i}-a_{i-1})^{2}\sup_{0\leq x\leq1}\left\{ \alpha_{i}x+\frac{\beta}{2}x^{2}-\frac{1}{2}I(x)\right\} .\nonumber
\end{align}
\end{proof}

The net benefit function $\alpha(x,y)$ assumed in the Proposition is shown in Figure \ref{fig:partitionsalpha}(D).
Essentially this result means that with extreme homophily,
we can approximate the model, assuming perfect segregation:
thus we can independently solve the variational problem of each type. This approach is computationally very simple, since each variational problem becomes a univariate maximization problem. \\
\indent The solution of such univariate problem has been studied and characterized in previous work by \cite{DiaconisChatterjee2011}, \cite{RadinYin2013}, \cite{AristoffZhu2014} and \cite{Mele2010a}. It can be shown that the solutions $\mu_{m}^{\ast}$, where $m=1,..,M$, are the fixed point of equations
\begin{equation}
\mu_{m} = \frac{\exp\left[\alpha_{mm}+\beta \mu_{m}\right]}{1+\exp\left[\alpha_{mm}+\beta \mu_{m}\right]}\,,
\end{equation}
for each group $m$, and $\beta\mu_{m}^{\ast}(1-\mu_{m}^{\ast})<1$. 
The global maximizer $\mu_{m}^{\ast}$ is unique except
on a phase transition curve $\{(\alpha_{mm},\beta):\alpha_{mm}+\beta=0,\alpha_{mm}<-1\}$, see e.g. \cite{RadinYin2013,AristoffZhu2014}. 
It is shown in \cite{DiaconisChatterjee2011} that the network of each group corresponds to an Erd\H{o}s-R\'{e}nyi graph with probability of a link equal to $\mu_{m}^{\ast}$.

%%%%%%%%%%%%%%%%%%%%%%%%%%%%%%%%%%%%%%%%%%
\subsection{Analytically Tractable Bounds}

In this section, for the edge-star model, we provide
analytically tractable bounds for $\psi(\alpha,\beta,\gamma)$
when $\gamma=0$.

\begin{proposition}
\label{prop:boundsConstant}
Let $\gamma=0$ and $0=a_{0}<a_{1}<\cdots<a_{M-1}<a_{M}=1$ be a given sequence. Let us assume that
\begin{equation*}
\alpha(x,y)=\alpha_{ml},\qquad\text{if $a_{m-1}<x<a_{m}$ and $a_{l-1}<y<a_{l}$, where $1\leq m,l\leq M$}.
\end{equation*}
Then, we have
\begin{align*}
&\sup_{\substack{0\leq u_{ml}\leq 1
\\
u_{ml}=u_{lm}, 1\leq m,l\leq M}}\sum_{m=1}^{M}(a_{m}-a_{m-1})
\bigg\{\sum_{l=1}^{M}(a_{l}-a_{l-1})\alpha_{ml}u_{ml}
\\
&\qquad
+\frac{\beta}{2}\left(\sum_{l=1}^{M}(a_{l}-a_{l-1})u_{ml}\right)^{2}
-\frac{1}{2}\sum_{l=1}^{M}(a_{l}-a_{l-1})I(u_{ml})\bigg\}
\nonumber
\\
&\leq
\psi(\alpha,\beta,0)
\leq\sum_{m=1}^{M}(a_{m}-a_{m-1})\sup_{\substack{0\leq u_{ml}\leq 1
\\
1\leq l\leq M}}
\bigg\{\sum_{l=1}^{M}(a_{l}-a_{l-1})\alpha_{ml}u_{ml}
+\frac{\beta}{2}\left(\sum_{l=1}^{M}(a_{l}-a_{l-1})u_{ml}\right)^{2}
\nonumber
\\
&\qquad\qquad\qquad\qquad
-\frac{1}{2}\sum_{l=1}^{M}(a_{l}-a_{l-1})I(u_{ml})\bigg\}.
\nonumber
\end{align*}
\end{proposition}

\begin{proof}
To compute the lower and upper bounds, let us define
\begin{equation}
u_{ij}(x)=\frac{1}{a_{j}-a_{j-1}}\int_{a_{j-1}}^{a_{j}}h(x,y)dy,\qquad\text{for any $a_{i-1}<x<a_{i}$}.
\end{equation}
We can compute that
\begin{equation}\label{s1}
\iint_{[0,1]^{2}}\alpha(x,y)h(x,y)dxdy
=\sum_{i=1}^{M}\sum_{j=1}^{M}(a_{j}-a_{j-1})\int_{a_{i-1}}^{a_{i}}\alpha_{ij}u_{ij}(x)dx.
\end{equation}
Moreover,
\begin{align}\label{s2}
\frac{\beta}{2}\int_{0}^{1}\int_{0}^{1}\int_{0}^{1}h(x,y)h(y,z)dxdydz
&=\frac{\beta}{2}\int_{0}^{1}\left(\int_{0}^{1}h(x,y)dy\right)^{2}dx
\\
&=\frac{\beta}{2}\sum_{i=1}^{M}\int_{a_{i-1}}^{a_{i}}\left(\sum_{j=1}^{M}(a_{j}-a_{j-1})u_{ij}(x)\right)^{2}dx.
\nonumber
\end{align}
By Jensen's inequality, we can also compute that
\begin{align}\label{s3}
\frac{1}{2}\int_{0}^{1}\int_{0}^{1}I(h(x,y))dxdy
&=\frac{1}{2}\sum_{i=1}^{M}\int_{a_{i-1}}^{a_{i}}\left[\sum_{j=1}^{M}\int_{a_{j-1}}^{a_{j}}I(h(x,y))dy\right]dx
\\
&=\frac{1}{2}\sum_{i=1}^{M}\int_{a_{i-1}}^{a_{i}}\left[\sum_{j=1}^{M}
(a_{j}-a_{j-1})\frac{1}{a_{j}-a_{j-1}}\int_{a_{j-1}}^{a_{j}}I(h(x,y))dy\right]dx
\nonumber
\\
&\geq\frac{1}{2}\sum_{i=1}^{M}\int_{a_{i-1}}^{a_{i}}\left[\sum_{j=1}^{M}
(a_{j}-a_{j-1})I\left(\frac{1}{a_{j}-a_{j-1}}\int_{a_{j-1}}^{a_{j}}h(x,y)dy\right)\right]dx
\nonumber
\\
&=\frac{1}{2}\sum_{i=1}^{M}\int_{a_{i-1}}^{a_{i}}\sum_{j=1}^{M}
(a_{j}-a_{j-1})I(u_{ij}(x))dx
\nonumber
\end{align}
Hence, by \eqref{s1}, \eqref{s2}, \eqref{s3}, we get
\begin{align}
\psi(\alpha,\beta,0)
&\leq
\sum_{i=1}^{M}\sum_{j=1}^{M}(a_{j}-a_{j-1})\int_{a_{i-1}}^{a_{i}}\alpha_{ij}u_{ij}(x)dx
+\frac{\beta}{2}\sum_{i=1}^{M}\int_{a_{i-1}}^{a_{i}}\left(\sum_{j=1}^{M}(a_{j}-a_{j-1})u_{ij}(x)\right)^{2}dx
\nonumber
\\
&\qquad\qquad
-\frac{1}{2}\sum_{i=1}^{M}\int_{a_{i-1}}^{a_{i}}\sum_{j=1}^{M}
(a_{j}-a_{j-1})I(u_{ij}(x))dx
\nonumber
\\
&\leq\sum_{i=1}^{M}(a_{i}-a_{i-1})\sup_{\substack{0\leq u_{ij}\leq 1
\\
1\leq j\leq M}}
\bigg\{\sum_{j=1}^{M}(a_{j}-a_{j-1})\alpha_{ij}u_{ij}
+\frac{\beta}{2}\left(\sum_{j=1}^{M}(a_{j}-a_{j-1})u_{ij}\right)^{2}
\nonumber
\\
&\qquad\qquad\qquad\qquad
-\frac{1}{2}\sum_{j=1}^{M}(a_{j}-a_{j-1})I(u_{ij})\bigg\}
\nonumber
\end{align}
On the other hand,
by restricting the supremum over the graphons $h(x,y)$
\begin{equation}
h(x,y)=
u_{ij},\qquad\text{if $a_{i-1}<x<a_{i}$ and $a_{j-1}<y<a_{j}$, where $1\leq i,j\leq M$},
\end{equation}
where $(u_{ij})_{1\leq i,j\leq M}$ is a symmetric matrix of the constants,
and optimize over all the possible values $0\leq u_{ij}\leq 1$, we get the lower bound:
\begin{align}
\psi(\alpha,\beta,0)
&\geq\sup_{\substack{0\leq u_{ij}\leq 1
\\
u_{ij}=u_{ji}, 1\leq i,j\leq M}}\sum_{i=1}^{M}(a_{i}-a_{i-1})
\bigg\{\sum_{j=1}^{M}(a_{j}-a_{j-1})\alpha_{ij}u_{ij}
\\
&\qquad
+\frac{\beta}{2}\left(\sum_{j=1}^{M}(a_{j}-a_{j-1})u_{ij}\right)^{2}
-\frac{1}{2}\sum_{j=1}^{M}(a_{j}-a_{j-1})I(u_{ij})\bigg\}.
\nonumber
\end{align}
\end{proof}

%%%%%%%%%%%%%%%%%%%%%%%%%%%%%%%%%%%%%%%%%%%%%%%%%%%%%%%%%%%%%%%%%%

%\bibliographystyle{jmr}
%\bibliography{thesisbib}

\end{document}